\documentclass[runningheads]{llncs}
\usepackage[T1]{fontenc}
\usepackage{graphicx}
\usepackage[numbers]{natbib}
\usepackage{amssymb}
\usepackage{amsmath}
\usepackage{amsfonts}
\usepackage{mathtools}
\usepackage{algpseudocode}
\usepackage{tikz}

\usepackage{booktabs} % For formal tables
\usepackage[linesnumbered,ruled]{algorithm2e} % For algorithms

\newcommand{\WPSlong}{\textsc{Weighted Proportional Sharing}\xspace}
\newcommand{\WPSshort}{\textsc{WPS}\xspace}
\newcommand{\WPSh}{\textsc{WPS}-h\xspace}
\newcommand{\WPShat}{\textsc{WPS}-$\hat{\mathrm{h}}$\xspace}
\newcommand{\WPSd}{\textsc{WPS}-d\xspace}
\newcommand{\WPSo}{\textsc{WPS}-o\xspace}
\newcommand{\WPSw}{\textsc{WPS}-w\xspace}
\newcommand{\hmin}{t}

\SetAlFnt{\small}
\SetAlCapFnt{\small}
\SetAlCapNameFnt{\small}
\SetAlCapHSkip{0pt}
\IncMargin{-\parindent}

\usepackage[a4paper, left=1in, right=1in, top=1in, bottom=1in]{geometry}

\begin{document}
\title{Achieving Coordination in Non-Cooperative\\ Joint Replenishment Games}

\author{Junjie Luo\inst{1} \and Changjun Wang\inst{2}}
\institute{School of Mathematics and Statistics, Beijing Jiaotong University \email{jjluo1@bjtu.edu.cn} \and
Academy of Mathematics and Systems Science, Chinese Academy of Sciences \email{wcj@amss.ac.cn}
}

\authorrunning{J. Luo and C. Wang}

\maketitle              % typeset the header of the contribution
\begin{abstract}
We analyze an infinite-horizon deterministic joint replenishment model from a non-cooperative game-theoretical approach. 
In this model, a group of retailers can choose to jointly place an order, which incurs a major setup cost independent of the group, and a minor setup cost for each retailer. Additionally, each retailer is associated with a holding cost.
Our objective is to design cost allocation rules that minimize the long-run average system cost while accounting for the fact that each retailer independently selects its replenishment interval to minimize its own cost.
We introduce a class of cost allocation rules that distribute the major setup cost among the associated retailers in proportion to their predefined weights. 
For these rules, we establish a monotonicity property of agent better responses, which enables us to prove the existence of a payoff dominant pure Nash equilibrium that can also be computed efficiently.
We then analyze the efficiency of these equilibria by examining the price of stability (PoS), the ratio of the best Nash equilibrium's system cost to the social optimum, across different information settings.
In particular, our analysis reveals that one rule, which leverages retailers’ own holding cost rates, achieves a near-optimal PoS of 1.25, while another rule that does not require access to retailers' private information also yields a favorable PoS.

\keywords{Joint replenishment \and Non-cooperative game \and Cost-allocation rules 
% \and Payoff dominant Nash equilibrium 
\and Price of Stability 
% \and Near optimal
}
\end{abstract}
\section{Introduction}

The joint replenishment problem (JRP) is one of the most studied problems in inventory theory, which extends the classical economic order quantity (EOQ) model \citep{harris1990many} from a single retailer to multiple retailers.
Like the basic EOQ model, each retailer in JRP faces a constant and deterministic market demand rate and incurs a holding cost per item per unit of time.
The new aspect of JRP is that, to reduce replenishment costs, a group of retailers can place a joint order together.
Every joint order incurs a fixed major setup cost which is independent of the group, and a minor setup cost for each participating retailer.
This joint setup cost structure is known as \emph{first-order interaction} \citep{federgruen1995efficient}.
The objective of JRP is to find a joint replenishment policy for all retailers to minimize the infinite-horizon average cost for the whole system.

Although the existence of an optimal joint replenishment policy for JRP has already been proved by \citet{adelman_duality_2005}, the structure of the optimal policy is still unknown and it is often very complex.
Consequently, research has focused on policies that are computationally efficient, easy to implement, and meanwhile close to optimal. 
In particular, the so called \emph{power-of-two} (POT) policies, in which each retailer's replenishment interval is required to be an integer power-of-two times a common base planning period, have been widely studied in the literature.
Besides being simple, it is well-known that the optimal POT policy has an average cost that is within 6\% (or 2\%) of a lower bound for the optimal cost when the base planning period is fixed in advance (or can be chosen freely) \citep{roundy198598,jackson1985joint}.
This result has also been extended to a general joint setup cost structure assuming only monotonicity and submodularity \citep{federgruen1992joint}.

Once the optimal POT policy is identified, a natural question is how should the total cost be distributed among all retailers in the system.
Since cost allocation is crucial for the long-term success of the system, researchers have leveraged game theory to analyze this problem from both cooperative and non-cooperative perspectives.
Focusing on the cooperative perspective, \citet{anily_cost_2007} model cost allocation as a cooperative game and aim to find a fair (or stable) cost allocation (the core of the game), under which no subgroup of retailers can benefit by forming a new coalition and using an optimal policy for themselves.
Specifically, they define a characteristic function that assigns each subset of retailers the cost of the optimal POT policy of the JRP instance restricted to that subset.
They show that this function is submodular, and hence the resulting cooperative game has a non-empty core.
\citet{zhang_cost_2009} and \citet{he2012polymatroid} extend these results to a more generalized joint setup cost that is nondecreasing and submodular.

Compared to the cooperative approach, the non-cooperative approach has received far less attention.
However, in decentralized systems where retailers pursue their own self-interest, binding cooperation is not always guaranteed.
While retailers can cooperate to reduce total costs, they may also strategically compete for individual benefits, which can result in significant efficiency losses and finally break the cooperation.
Therefore, it is essential to account for the \emph{strategic behavior} of retailers when designing cost allocation rules.
Moreover, cooperative core allocations frequently rely on detailed cost information from retailers, which in decentralized systems may be private information and thus unavailable for allocation rules.
Hence, it is also important to consider \emph{private information} constraints.

\citet{he_noncooperative_2017} introduce the following non-cooperative game model for the cost allocation problem in joint replenishment:
Each retailer independently selects its own replenishment interval, pays its own minor setup cost and holding cost, and shares the major setup costs according to a preannounced rule.
Unlike the cooperative approach, this mode does not require all retailers to make binding agreements on a joint replenishment policy, making it more applicable in decentralized settings.
Under this model, \citet{he_noncooperative_2017} study a natural rule, called the equal-division rule, where the major setup cost of each joint order is evenly shared by all participating retailers.
This rule does not rely on private cost information from retails and is easy to implement.
They prove that under this rule, there exists a payoff dominant (pure) Nash equilibrium which achieves the optimal payoff for every retailer among all Nash equilibria.

However, under the equal-division rule,
the selfish behavior of retailers optimizing their own costs can lead to a significantly higher total system cost.
Indeed, 
\citet{he_noncooperative_2017} demonstrate that the price of anarchy (PoA) \citep{koutsoupias1999worst}, the ratio between the system cost at the \emph{worst} Nash equilibrium and the optimal system cost under centralized coordination, is $O(\sqrt{n})$, where $n$ is the number of retailers.
Furthermore, even the price of stability (PoS) \citep{schulz2003performance,anshelevich2008price}, which quantifies the efficiency loss at the \emph{best} Nash equilibrium, remains as large as $O(\sqrt{\ln n})$.
These results highlight the significant inefficiency introduced by decentralized decision-making under the equal-division rule.
A central open question left by \citet{he_noncooperative_2017} is whether alternative cost allocation rules can improve efficiency in a decentralized system and achieve the coordination effect as in a centralized cooperative system.

\subsection{Our Contributions}

In this paper, we conduct a systematic study of cost allocation rules for the non-cooperative joint replenishment game, where each retailer independently selects its replenishment interval.
Our goal is to identify practical allocation rules that achieve efficient coordination and guarantee strong social performance in decentralized systems.
We structure our contributions around the following three guiding questions, which are addressed in Sections \ref{sec:WPS}, \ref{sec:efficient-rules}, and \ref{sec:H-unkown}, respectively.

% \vspace{.05cm}
\begin{quote}
    \textbf{\emph{Question 1: Can we design cost allocation rules, beyond the equal-division rule, that guarantee the existence of a payoff-dominant Nash equilibrium?
    }}
\end{quote}
% \vspace{.05cm}

To address Question 1, we introduce a broad class of cost allocation rules, called \WPSlong (\WPSshort) rules, which generalizes the equal-division rule by assigning each retailer a weight.
Under these rules, the major setup cost of each joint order is shared proportionally to these weights among participating retailers. 
The equal-division rule corresponds to the special case where all retailers have equal weights.
We establish a key monotonicity property for all \WPSshort rules regarding agents' better-response dynamics: If a retailer can reduce its cost by doubling its replenishment interval at a given state, then it has no incentive to revert this decision in future steps, provided that no other retailer subsequently decreases their interval.
Crucially, this monotonicity enables us to prove the existence of a payoff dominant \citep{harsanyi1988general} Nash equilibrium for every \WPSshort rule, where each retailer achieves its highest possible payoff among all Nash equilibria.
Moreover, this equilibrium can be computed efficiently via a simple and natural polynomial-time algorithm.
Our results generalize the findings of \citet{he_noncooperative_2017} for the equal-division rule.
Notably, the equal-division result follows from the submodularity of the corresponding game, whereas our broader results are obtained via a deeper structural understanding of agents' better-response behaviors.

Having established the existence and efficient computation of payoff-dominant equilibria for all \WPSshort rules, we next focus on designing simple and practical \WPSshort rules that achieve strong coordination performance.
Since a payoff-dominant equilibrium ensures the optimal payoff for every retailer among all Nash equilibria, it is widely adopted as a criterion for choosing among equilibrium points in non-cooperative games~\citep{colman1997payoff}.
Therefore, we use the price of stability (PoS), which precisely quantifies the efficiency loss of the payoff dominant Nash equilibrium, as our primary performance criterion.
This leads us to our second research question:

% \vspace{.05cm}
\begin{quote}
    \textbf{\emph{Question 2: Are there simple and practical \WPSshort rules that achieve near-optimal coordination, as measured by the price of stability?}}
\end{quote}
% \vspace{.05cm}

To address Question 2, we first introduce \WPSo, a theoretical benchmark from the \WPSshort class that achieves an optimal PoS of exactly 1, demonstrating that perfect coordination efficiency is achievable in principle. 
However, despite its theoretical optimality, \WPSo has significant drawbacks that severely limit its practical use. 
Specifically, it introduces free-riders who contribute nothing to the major setup costs, and the weight assigned to each remaining retailer is constructed in a complex manner that depends not only on its own cost parameters but also on %the full information 
other retailers' cost parameters. 
This dependence on global information makes \WPSo unsuitable for decentralized environments.

Motivated by these limitations, we propose a new practical rule, called \WPSh, which allocates the major setup cost among all participating retailers proportionally to $H_i=\frac{1}{2}h_id_i$, where $h_i$ is retailer $i$'s holding cost rate per item per unit of time and $d_i$ is its demand rate.
Importantly, \WPSh eliminates free-riders entirely and assigns each retailer a simple, transparent weight depending only on its own parameters.
We show that the payoff dominant Nash equilibrium under \WPSh closely approximates the optimal centralized policy, in the sense that each retailer either %holds the optimal replenishment interval 
adopts the same replenishment interval as in the optimal centralized solution or doubles it.
Moreover, this equilibrium can be computed efficiently in linear time.
Consequently, we prove that the PoS of \WPSh is at most 1.25, demonstrating that \WPSh achieves near-optimal coordination performance.
Complementing this positive result, we establish a fundamental lower bound showing that, no \WPSshort rule can achieve a PoS of 1 without access to minor setup costs $K_i$ (the best possible PoS is at least 1.05).

While \WPSh successfully resolves the practical drawbacks of \WPSo and achieves near-optimal efficiency, it still requires access to retailers' holding cost information $H_i$, which may be  private information in decentralized environments.
This naturally raises our final research question:

% \vspace{.05cm}
\begin{quote}
    \textbf{\emph{Question 3: What is the best achievable efficiency when all cost parameters are private, and how should we allocate costs under such information constraints?
    % How should we allocate cost when all cost parameters are private information?
    }}
\end{quote}
% \vspace{.05cm}

To address Question 3, we propose another \WPSshort rule, called \WPSd, which allocates the major setup cost proportionally to retailers' demand rates~$d_i$.
Notice that the demand rate of each retailer can be directly inferred from its replenishment interval~$T_i$ and order quantity~$d_iT_i$, both of which are observable and verifiable.
Thus, \WPSd does not reply on any private information.
To analyze the performance of \WPSd, we first establish a general upper bound on the PoS applicable to all \WPSshort rules.
From this result, we conclude that the PoS under \WPSd is $O(\sqrt{\log \gamma_d})$, where $\gamma_d$ is the ratio of the maximum to the minimum holding cost rates among all retailers.
Since holding cost rates typically vary moderately in practice, \WPSd guarantees strong performance in decentralized systems without requiring any private information. 

Moreover, we show that \WPSd is essentially optimal when cost information is private.
Specifically, we prove a matching lower bound: When holding cost information $H_i$ is private, no \WPSshort rule can achieve a PoS better than $\Omega(\sqrt{\log \gamma_d})$, even if all minor setup costs $K_i$ are zero.
This result reveals a fundamental distinction between the roles of the holding costs $H_i$ and the minor setup costs $K_i$:
When $K_i$ is private, \WPSh still achieves a near-optimal PoS of 1.25, whereas when $H_i$ is private, no \WPSshort rule can achieve a PoS below $\Omega(\sqrt{\log \gamma_d})$.

Recognizing the critical importance of holding cost information, we also explore a practical scenario where $H_i$ is not exactly known but can be reasonably estimated.
Building on our general PoS bound, we propose a robust \WPSshort rule based on estimated values of $H_i$, which achieves good coordination performance as long as the estimation errors remain moderate.

In summary, our work provides a comprehensive understanding of cost allocation in non-cooperative joint replenishment games by (1) characterizing a broad class of rules that guarantee the existence of payoff-dominant Nash equilibria, (2) establishing fundamental efficiency limits under varying information constraints, and (3) designing practical cost allocation rules that achieve or closely approximate optimal coordination in decentralized settings.

\subsection{Further Related Work}
\label{sec:related-work}

\paragraph{Joint Replenishment Problem.} 
A key feature of the game model studied by \citet{he_noncooperative_2017}, which we also examine in this paper, is that each retailer independently chooses its own replenishment interval. In contrast, the models discussed below assume a uniform replenishment interval for all retailers, typically determined by an intermediary.
\citet{meca_cooperation_2003} consider a non-cooperative game for a special case of the joint replenishment model, where there is no minor setup cost. 
In this game, every retailer submits its optimal number of orders per unit time if it orders alone.
Based on this information a final number of joint orders per unit of time is computed.
Then the allocation rule of \citet{meca_inventory_2004} is applied to allocate the major setup costs and each retailer pays its own holding cost.
\citet{meca_cooperation_2003} provides a necessary and sufficient condition for the existence of a constructive equilibrium (a Nash equilibrium where every retailer submits a positive number) in which all retailers make joint orders.  
\citet{korpeoglu_non-cooperative_2013} also study a non-cooperative joint replenishment model with zero minor setup costs.
In this game, each retailer reports to an intermediary how much it is willing to pay for replenishment costs per unit of time and the intermediary then determines the maximum ordering frequency for the joint orders that can be financed with retailers' contributions.
They show the existence of a pure-strategy Bayesian Nash equilibrium in this game, which can be characterized by a system of integral equations. 

Next, we turn to related work on cooperative models.
\citet{meca_inventory_2004} study a special case of the cooperative model of \citet{anily_cost_2007} where there is no minor setup cost.
They propose a cost-allocation rule in which each retailer’s share of the total cost is proportional to the square of its optimal number of orders per unit of time, and show that this rule leads to a core allocation.
Interestingly, this allocation rule is equivalent to our proposed \WPSh, although our game models are different.
For the same model as \citet{anily_cost_2007} but with the added restriction that retailers always place orders together, \citet{dror_shipment_2007} provide a sufficient and necessary condition for the non-emptiness of the core.
\citet{dror_cost_2012} further analyze the sensitivity of the core with respect to the changes in  cost parameters.

\paragraph{Network Cost-Sharing Games.}
Network cost-sharing games have been extensively studied in the field of algorithmic game theory.
\citet{anshelevich2008price} introduce a network cost-sharing game, where we are given a graph with fixed edge costs and $k$ players each having a source-sink pair.
Each player chooses a path in the graph to connect its source-sink pair.
The choices made by all players induce a network and a social cost, which is the sum of the cost of edges in the network.
\citet{anshelevich2008price} show that a pure-strategy Nash equilibrium always exists under the equal-division rule, i.e., the cost of each used edge is equally shared by agents who use this edge, and the PoS under this rule is $O(\log k)$.
The existence result is proved by using the potential function method due to \citet{monderer1996potential}.
By contrast, the existence of a pure-strategy Nash equilibrium in the JRP game under the equal-division rule follows by the submodularity of the game~\citep{he_noncooperative_2017}, and our more general existence result for all \WPSshort rules reply on structural properties of the game.

\citet{anshelevich2008price} also study a network cost-sharing game with weighted players under the rule that the cost of each used edge is shared by agents who use this edge in proportional to their weights, which is equivalent to our \WPSshort rules.
However, the weights in their work are given explicitly in the input to reflect different amounts of traffic of players, whereas the weights used in our \WPSshort rules are designed based on different cost parameters of retailers aiming to optimize the PoS.
Moreover, \citet{anshelevich2008price} show that the potential function technique cannot be directly used for weighted players. 
Indeed, as shown by \citet{ChenR09}, a pure-strategy Nash equilibrium is not guaranteed to exist for the network cost-sharing game with weighted players and consequently they study approximate Nash equilibrium, whereas we show that a payoff dominant pure-strategy Nash equilibrium always exists in our JRP game under any \WPSshort rule.

In the same spirit as us, \citet{ChenRV10} study the problem of designing better cost-sharing rules to improve the equilibrium efficiency of the resulting network cost-sharing game and show that different information assumptions about the designer can lead to rules with different equilibrium efficiency.

\section{Preliminaries}

\paragraph{First-Order Interaction Joint Replenishment Model.}
There is a set $N=\{1,2,\dots,n\}$ of retailers, each has a fixed market demand rate $d_i>0$. 
To reduce replenishment costs, retailers can place a joint order from the single warehouse to satisfy their demands. 
For a subset $S \subseteq N$ of retailers, the joint setup cost is
$K^0(S)=K_0+\sum_{i \in S}K_i$.
% \[
% K^0(S)=K_0+\sum_{i \in S}K_i,
% \]
where $K_0$ is the fixed major setup cost independent of the set $S$ of retailers and their orders, and $K_i$ is the minor setup cost for retailer $i$.
In addition, each retailer $i \in N$ has a holding cost rate $h_i$ per item per-unit of time.
We assume zero lead times and no backlogging or shortage. 
Then the replenishment policy of each retailer is determined by its replenishment interval.
For a retailer $i \in N$ with a replenishment interval $T_i$, the order quantity is $d_iT_i$ and the holding cost per-unit of time is $\frac{1}{2}h_id_iT_i$. For convenience, we denote $H_i=\frac{1}{2}h_id_i$.
For any subset $S \subseteq N$ of retailers, denote $H(S)=\sum_{i \in S} H_i$ and $K(S)=\sum_{i \in S} K_i$.

\paragraph{Optimal Centralized Policy.}
Denote $T=(T_i:i\in N)$ the joint replenishment policy.
In this paper, we consider the class of power of two policies \citep{roundy198598,jackson1985joint}. That is, there is a common base planning period $B$, and each retailer's replenishment interval is required to be
$
T_i=2^{z_i}B
$,
where $z_i \in \mathbb{Z}$ is an integer.
Given a joint replenishment policy $T$, denote $T_{\min}=\min_{i \in N}T_i$, 
then the system’s average cost per unit of time is computed as
$C(T)=\sum_{i \in N}\left(H_iT_i+\frac{K_i}{T_i}\right)+\frac{K_0}{T_{\min}}$.
% \[
% C(T)=\sum_{i \in N}\left(H_iT_i+\frac{K_i}{T_i}\right)+\frac{K_0}{T_{\min}}.
% \]
% Jackson et al. proved that the optimal solution to the centralized system $T^c$ can be computed as follows.
Assume that $\frac{K_1}{H_1} \leq \frac{K_2}{H_2} \leq \dots \leq \frac{K_n}{H_n}$.
Let $i^*=\max\{j: 1 \leq j \leq n \text{ such that } \frac{K_0+\sum_{1 \le i \le j}K_i}{\sum_{1 \le i \le j}H_i} \geq \frac{K_j}{H_j}\}$.
Denote $U=\{1,2,\dots,i^*\}$ and $V=N \setminus U$.
\citet{jackson1985joint} prove that in the optimal centralized policy~$T^c$, all agents in $U$ order together with the same replenishment interval, while each retailer in $V$ orders its own EOQ as if there is no major setup cost.
Formally,
\begin{equation}
\label{eq:opt}
T_i^c \overset{POT}{=}
\begin{cases}
\sqrt{s}, \text{ if } i \in U,\\
\sqrt{\frac{K_i}{H_i}}, \text{ if } i \in V,
\end{cases}
\end{equation}
where
\begin{align}\label{eq:s-def}
s=\frac{K_0+\sum_{j \in U}K_j}{\sum_{j \in U}H_j} = \min_{\emptyset \subset S \subseteq N}\frac{K_0+K(S)}{H(S)}.
\end{align}
Here we use the notation $a\overset{POT}{=}b$ to denote that $a$ is the multiplicatively nearest POT value to $b$, i.e., $a=2^{\lfloor \log_2(b)+0.5\rfloor}$. By definition we have $a \in (b/\sqrt{2}, \sqrt{2}b]$.
Note that sets $U$ and $V$ can also be characterized as follows:
\begin{align}\label{eq:U-V}
U=\{i \in N \mid K_i \leq sH_i\} \text { and }
V=\{i \in N \mid K_i > sH_i\}.
\end{align}

\paragraph{Non-cooperative Joint Replenishment Game \citep{he_noncooperative_2017}.}
The set of players is the set of retailers~$N$. Each retailer $i$ decides its own replenishment interval $T_i$. 
% Denote $T=(T_i:i\in N)$ the joint replenishment policy.
% In this paper, we consider the class of power of two policies. That is, there is a common base planning period $B$, and each retailer's replenishment interval is required to be
% \[
% T_i=2^{z_i}B,
% \]
% where $z_i \in \mathbb{Z}$ is an integer.
For each retailer $i \in N$, if it pays the whole $K_0$, then its overall cost is $H_iT_i+\frac{K_0 + K_i}{T_i}$ and the corresponding optimal replenishment interval is $\sqrt{(K_0+K_i) / H_i}$. If it does not share $K_0$ at all, then its optimal replenishment interval is $\sqrt{K_i / H_i}$.
Since we consider POT policies,
we can restrict the set of strategies of retailer $i$ by 
% $\Gamma_i=\{T_i \mid \sqrt{K_i / (2H_i)} \le T_i \le \sqrt{2(K_0+K_i) / H_i} \text{ such that } T_i=2^{z_i}B, z_i \in \mathbb{Z}\}$.
\[\Gamma_i=\{T_i \mid \sqrt{K_i / (2H_i)} \le T_i \le \sqrt{2(K_0+K_i) / H_i} \text{ such that } T_i=2^{z_i}B, z_i \in \mathbb{Z}\}.
\]
Notice that there exists at least one POT point in $\Gamma_i$.
Let $\Gamma=\times_{i \in N} \Gamma_i$ be the strategy profile set for all retailers and $\Gamma_{-i}=\times_{j \in N\setminus\{i\}} \Gamma_j$ be the strategy profile set for all retailers except for $i$.

Given the submitted joint replenishment policy, each retailer pays its own minor setup cost and holding cost, and the major setup cost is shared by retailers according to a preannounced rule.
% \paragraph{Cost Sharing Rule.}
Our goal is to design a cost allocation rule for the game such that the total cost of the system approaches that under the optimal centralized policy even if we allow retailers to strategically decide their own replenishment intervals.
Formally, denote~$\mathcal{P}=(K_0,N,\{K_i\}_{i \in N},\{h_i\}_{i \in N},\{d_i\}_{i \in N})$ the profile of the system.
A cost allocation rule is a function $x$ that maps a profile $\mathcal{P}$ and a joint replenishment policy $T$ to an allocation scheme $x(T)$ of the major setup cost per unit of time $K_0/T_{\min}$ to retailers.
Naturally, we require that $x_i(T) \geq 0$ for each $i \in N$.
The long-run average cost of retailer~$i$ under cost allocation rule~$x$ is
$f_i(T)=H_iT_i+\frac{K_i}{T_i}+x_i(T)$.
% \[
% f_i(T)=H_iT_i+\frac{K_i}{T_i}+x_i(T).
% \]
% $T$ and a public information set $I \subseteq \{K_0\} \cup \{K_i\}_{i \in N} \cup \{h_i\}_{i \in N} \cup \{d_i\}_{i \in N}$ to a cost allocation vector.
Given a cost allocation rule $x$, denote the induced game by $\mathcal{G}=(\mathcal{P},x)$.

\paragraph{Nash Equilibrium and Efficiency.}
A joint replenishment policy $T$ is called a (pure-strategy) Nash equilibrium under a cost allocation rule $x$ if no retailer can benefit by unilaterally changing its replenishment interval, i.e., $f_i(T_i;T_{-i}) \leq \min_{T'_i \in \Gamma_i}f_i(T'_i;T_{-i}), \forall i \in N, \forall T'_i \in \Gamma_i$.
% \[
% f_i(T_i;T_{-i}) \leq \min_{T'_i \in \Gamma_i}f_i(T'_i;T_{-i}), \forall i \in N, \forall T'_i \in \Gamma_i.
%  % \forall T_{-i} \in \Gamma_{-i}.
% \]
Let $\text{NE}(\mathcal{G})$ denote the set of all Nash equilibria of a game~$\mathcal{G}$ and let $T^c(\mathcal{G})$ denote the optimal centralized policy for $\mathcal{G}$.
Given a cost allocation rule $x$, the price of anarchy $\text{PoA}(x)$ (resp. price of stability $\text{PoS}(x)$) is defined as the worst-case ratio over all induced games~$\mathcal{G}=(\mathcal{P},x)$ between the \emph{worst} (resp. \emph{best}) equilibrium cost and the optimal cost:
\[
\text{PoA}(x)=\sup_{\mathcal{G}} \frac{\max_{T \in \text{NE}(\mathcal{G})}C(T)}{C(T^c(\mathcal{G}))}
\text{ and }
\text{PoS}(x)=\sup_{\mathcal{G}} \frac{\min_{T \in \text{NE}(\mathcal{G})}C(T)}{C(T^c(\mathcal{G}))}.
\]
% The price of stability $\text{PoS}(x)$ is defined as the worst-case ratio over all induced games~$\mathcal{G}=(\mathcal{P},x)$ between the \emph{best} equilibrium cost and the optimal cost:
% \[
% \text{PoS}(x)=\sup_{\mathcal{G}} \frac{\min_{T \in \text{NE}(\mathcal{G})}C(T)}{C(T^c(\mathcal{G}))}.
% \]
A joint replenishment policy $T$ is called a payoff dominant Nash equilibrium if it is a Nash equilibrium and 
it minimizes the cost for every retailer among all Nash equilibria.
Clearly, a payoff dominant Nash equilibrium achieves the lowest overall system cost among all Nash equilibria, making PoS a measure of the efficiency loss associated with the payoff dominant Nash equilibrium.

\section{Weighted Proportional Sharing Rules}
\label{sec:WPS}

In this section, we introduce a broad class of cost allocation rules for decentralized joint replenishment, which we call \WPSlong (\WPSshort) rules.
We begin by defining \WPSshort rules (Section~\ref{sec:WPS-intro}) and analyzing their structural properties (Section~\ref{sec:WPS-property}).
Based on these properties, we show that every \WPSshort rule guarantees the existence of a payoff-dominant Nash equilibrium (Section~\ref{sec:WPS-NE}).
Finally, we motivate the use of PoS as the criterion for evaluating the efficiency of different rules (Section~\ref{sec:WPS-efficiency}).

\subsection{Introducing \WPSshort Rules}
\label{sec:WPS-intro}

We start by formally defining the class of \WPSshort rules.
% , which are parameterized by a weight vector~$w=(w_i)_{i \in N}$.

\begin{definition}
% (\WPSshort rules)
A cost allocation rule is called a \WPSlong (\WPSshort) rule, if for any profile $\mathcal{P}$ there exists a weight vector~$w=(w_i \geq 0)_{i \in N}$ such that the major setup cost is allocated as follows:
Whenever a subset $S \subseteq N$ of retailers place a joint order, the corresponding major setup cost is shared by retailers in $S$ proportionally to their weights $w_i$. (For joint orders including only retailers with $w_i=0$ we allocate the major setup cost equally among these retailers.)
% \begin{itemize}
%   \item Every retailer pays its own minor setup cost and holding cost;
%   \item When a joint order is placed, all ordering retailers share the major setup cost proportionally to their weights $w_i$.\footnote{For joint orders including only retailers with $w_i=0$ we allocate the major setup cost equally among these retailers.}
% \end{itemize} 
\end{definition}

The equal-division rule corresponds to the special case with $w_i=1$ for each $i \in N$ for any profile~$\mathcal{P}$.
Note that the weight vector~$w$ in the above definition can be related to the parameters in the profile~$\mathcal{P}$ but has to be independent of retailers' actions (their chosen replenishment intervals).
For example, \citet{he_noncooperative_2017} analyze a rule called Proportional Sharing Rule (PSR), where retailers share the major setup cost proportionally to their order quantity $d_iT_i$.
Since the order quantity $d_iT_i$ relies on retailers' actions $T_i$, which is not prefixed for each profile $\mathcal{P}$,
PSR does not belong to the class of \WPSshort rules. 
As shown by \citet{he_noncooperative_2017}, Nash equilibrium is not guaranteed to exist under PSR.
In contrast, we will show the existence of Nash equilibrium for all \WPSshort rules. 

For a joint replenishment policy $T$, denote $N[T_i; T_{-i}]$ the set of retailers with replenishment interval at most $T_i$ under $T$, i.e., $N[T_i; T_{-i}]=\{j \in N \mid T_j \leq T_i\}$.
Without loss of generality, assume that $T_1 \leq T_2 \leq \dots \leq T_n$.
Since we consider POT policies, we have a nice property that for every joint order involving retailer $i$, it also includes all retailers whose replenishment interval does not exceed $T_i$, which are all the retailers from $\{1,2,\dots,i-1\}$.
Therefore, for any joint order the set of associated retailers can only be one of the $n$ cases: $\{1\},\{1,2,\},\dots,\{1,2,\dots,n\}$.
It follows that for any set $\{1,2,\dots,i\}$, the corresponding order frequency is $1/T_i-1/T_{i+1}$, where we set $1/T_{n+1} \coloneqq 0$.
Under \WPSshort rules, for each joint order associated with retailers from $\{1,2,\dots,i\}$, the major setup cost $K_0$ is shared by all retailers from this set proportionally to their weights, i.e., retailer $\ell \in \{1,2,\dots,i\}$ pays $\frac{w_\ell}{\sum_{j \in \{1,2,\dots,i\}}w_j}K_0$ for the major setup cost.
Therefore, the long-run average cost of retailer $i$ under a \WPSshort rule is
\[
f_i(T_i;T_{-i})=H_iT_i+\frac{K_i}{T_i}+x_i(T_i;T_{-i}),
\]
where 
\[
x_i(T_i;T_{-i})=\sum_{m=i}^{n}\left(\frac{1}{T_m}-\frac{1}{T_{m+1}}\right)\frac{w_i}{\sum_{j \in N[T_m; T_{-m}]}w_j}K_0.
\]

Although the above computation for the sharing of the major setup cost looks complex, the change in this part when a retailer doubles or halves its replenishment interval is easy to compute.
Indeed, when retailer $i$ doubles its replenishment interval, it skips the joint orders together with all retailers from $N[T_i; T_{-i}]$ that are placed at time points in $\{(2j+1)T_i\mid j=0,1,2,\dots\}$.
So the reduction in the major setup cost for retailer $i$ is
$x_i(T_i;T_{-i})-x_i(2T_i;T_{-i})=\frac{w_i}{\sum_{j \in N[T_i; T_{-i}]}w_j}\frac{K_0}{2T_i}$.
% \[
% x_i(T_i;T_{-i})-x_i(2T_i;T_{-i})=\frac{w_i}{\sum_{j \in N[T_i; T_{-i}]}w_j}\frac{K_0}{2T_i}.
% \]
Similarly, the increase for halving the replenishment interval is
$x_i(T_i/2;T_{-i})-x_i(T_i;T_{-i})=\frac{w_i}{\sum_{j \in N[T_i/2; T_{-i}]}w_j}\frac{K_0}{T_i}$.
% \[
% x_i(T_i/2;T_{-i})-x_i(T_i;T_{-i})=\frac{w_i}{\sum_{j \in N[T_i/2; T_{-i}]}w_j}\frac{K_0}{T_i}.
% \]

\begin{figure}[tb]
% \resizebox{0.8\textwidth}{!}{
\begin{tikzpicture}[xscale=1.4]

% Draw number line (x-axis)
\draw[->] (0,0) -- (9,0) node[right] {Time};

% Draw points at 1, 2, 4, 8
\foreach \x in {0,1,2,3,4,5,6,7,8} {
    \draw (\x,0.1) -- (\x,0); % Tick marks
    \node[below] at (\x,-0.1) {\x}; % Numbers on the axis
}

% Label sets above points
\node[above] at (1,0.1) {\{1\}};
\node[above] at (2,0.1) {\{1, 2\}};
\node[above] at (3,0.1) {\{1\}};
\node[above] at (4,0.1) {\{1, 2\}};
\node[above] at (5,0.1) {\{1\}};
\node[above] at (6,0.1) {\{1, 2\}};
\node[above] at (7,0.1) {\{1\}};
\node[above] at (8,0.1) {\{1, 2, 3\}};

\end{tikzpicture}
% }
\caption{An example of three retailers with $T_1=1,T_2=2,T_3=8$.}
\label{fig:example}
\end{figure}
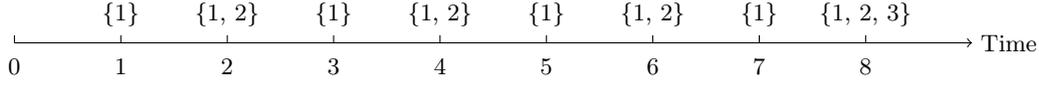

Consider the example in Figure~\ref{fig:example}, where there are three retailers with replenishment intervals $T_1=1,T_2=2,T_3=8$.
Retailer 2 participates in two different joint order sets: $\{1,2\}$, with order frequency of $\frac{1}{2}-\frac{1}{8}=\frac{3}{8}$, and
$\{1,2,3\}$, with order frequency of $\frac{1}{8}$.
Thus, the long-run average major setup cost for retailer 2 is
$\frac{3}{8}\frac{w_2}{w_1+w_2}K_0+\frac{1}{8}\frac{w_2}{w_1+w_2+w_3}K_0$.
If retailer 2 doubles its replenishment interval from 2 to 4, then it will skip the joint orders with retailer 1 that occur at time points  $\{2,6,10,14,\dots\}$ with frequency $\frac{1}{4}=\frac{1}{2T_2}$.
This reduces its major setup cost by $\frac{w_2}{w_1+w_2}\frac{K_0}{4}$.
Conversely, if retailer 2 halves its replenishment interval from 2 to 1, then it participate in additional joint orders with retailer 1 at time points $\{1,3,5,7,\dots\}$ with frequency $\frac{1}{2}=\frac{1}{T_2}$.
This increases its major setup cost by $\frac{w_2}{w_1+w_2}\frac{K_0}{2}$.

\subsection{Structural Properties}
\label{sec:WPS-property}

% \todo[inline]{It seems possible to obtain all the following results by showing the correspondence between \WPSshort and the weighted Shapley value, similar to OR'17.}

To build toward our main equilibrium results, we begin by examining key structural properties of \WPSshort rules.
We first characterize the condition under which a retailer can reduce its cost by increasing its replenishment interval.

\begin{lemma}
\label{lem:right-jump}
For any \WPSshort rule and any joint policy $T$, the following statements are equivalent:
\begin{enumerate}
\item Retailer $i$ can reduce its cost by increasing its replenishment interval to a certain extent;
\item Retailer $i$ can reduce its cost by doubling its replenishment interval;
\item It holds that
    \[
    T_i<\sqrt{\frac{1}{2}\frac{K_i+\frac{w_i}{\sum_{j \in N[T_i; T_{-i}]}w_j}K_0}{H_i}}.
    \]
\end{enumerate}
\end{lemma}

\begin{proof}
When retailer $i$ increases its replenishment interval from $T_i$ to $2T_i$, its holding cost is increased by $H_iT_i$ while its setup cost is decreased by 
\[
\frac{K_i+\frac{w_i}{\sum_{j \in N[T_i;T_{-i}]}w_j}K_0}{2T_i}.
\]
This change reduces retailer $i$'s total cost if and only if
\[
\frac{K_i+\frac{w_i}{\sum_{j \in N[T_i;T_{-i}]}w_j}K_0}{2T_i}>H_iT_i
\Leftrightarrow 
T_i< \sqrt{\frac{1}{2}\frac{K_i+\frac{w_i}{\sum_{j \in N[T_i;T_{-i}]}w_j}K_0}{H_i}}.
\]
% or
% \[
% T_i< \sqrt{\frac{1}{2}\frac{K_i+\frac{w_i}{\sum_{j \in N[T_i;T_{-i}]}w_j}K_0}{H_i}}.
% \]
Thus statements (2) and (3) are equivalent. 

Since (2) directly implies (1), it remains to show that (1) implies (2).
Suppose that retailer $i$ can reduce its cost by increasing its replenishment interval from $T_i$ to $2^{z}T_i$ for some $z \in \mathbb{Z^+}$.
After this change, the holding cost of retailer $i$ is increased by $(2^{z}-1)H_iT_i$ while the setup cost is decreased by 
% \[
% \sum_{x\in\{0,1,2,\dots,z-1\}}
% \frac{K_i+\frac{w_i}{\sum_{j \in N[2^xT_i;T_{-i}]}w_j}K_0}{2 \cdot 2^x \cdot T_i}
% \leq \sum_{x\in\{0,1,2,\dots,z-1\}}
% \frac{K_i+\frac{w_i}{\sum_{j \in N[T_i;T_{-i}]}w_j}K_0}{2 \cdot 2^x \cdot T_i}
% \leq 
% \left(1-\frac{1}{2^z}\right)\frac{K_i+\frac{w_i}{\sum_{j \in N[T_i;T_{-i}]}w_j}K_0}{T_i}.
% \]
\begin{align*}
 \sum_{x\in\{0,1,2,\dots,z-1\}}
 \frac{K_i+\frac{w_i}{\sum_{j \in N[2^xT_i;T_{-i}]}w_j}K_0}{2 \cdot 2^x \cdot T_i} 
& \leq   \sum_{x\in\{0,1,2,\dots,z-1\}} 
\frac{K_i+\frac{w_i}{\sum_{j \in N[T_i;T_{-i}]}w_j}K_0}{2 \cdot 2^x \cdot T_i} \\
 & \leq
\left(1-\frac{1}{2^z}\right)\frac{K_i+\frac{w_i}{\sum_{j \in N[T_i;T_{-i}]}w_j}K_0}{T_i}.
\end{align*}
Then
\[
\left(1-\frac{1}{2^z}\right)\frac{K_i+\frac{w_i}{\sum_{j \in N[T_i;T_{-i}]}w_j}K_0}{T_i}>(2^z-1)H_iT_i, 
\]
or
\begin{align*}
    T_i  <\sqrt{\frac{1}{2^z}\frac{K_i+\frac{w_i}{\sum_{j \in N[T_i;T_{-i}]}w_j}K_0}{H_i}}  \leq \sqrt{\frac{1}{2}\frac{K_i+\frac{w_i}{\sum_{j \in N[T_i;T_{-i}]}w_j}K_0}{H_i}},
\end{align*}
which means that retailer $i$ can also reduce its cost by doubling its replenishment interval according to the equivalence of statements (2) and (3).
Thus, (1) implies (2).
% When retailer $i$ increases its replenishment interval from $T_i$ to $2^{z}T_i$ for some $z \in \mathbb{Z^+}$, the holding cost of retailer $i$ is increased by $(2^z-1)H_iT_i$ while the setup cost is decreased by 
% \[
% \frac{K_i+\frac{H_i}{\sum_{j \in N[T_i;T_{-i}]}H_j}K_0}{2^zT_i}.
% \]
% This change reduces retailer $i$'s total cost if and only if
% \[
% \frac{K_i+\frac{H_i}{\sum_{j \in N[T_i;T_{-i}]}H_j}K_0}{2^zT_i}>(2^z-1)H_iT_i, 
% \]
% or
% \[
% T_i<\sqrt{\frac{1}{2^{2z}-2^z}(\frac{K_i}{H_i}+\frac{1}{\sum_{j \in N[T_i;T_{-i}]}H_j}K_0)}
% \leq \sqrt{\frac{1}{2}(\frac{K_i}{H_i}+\frac{1}{\sum_{j \in N[T_i;T_{-i}]}H_j}K_0)}.
% \]
% It immediately follows that statements (2) and (3) are equivalent and (1) implies that (2).
% Combining with the fact that (2) directly implies (1), we have shown  the equivalence for all the three statements.
\hfill$\square$ \end{proof}

According to Lemma \ref{lem:right-jump}, except for the information about retailer $i$, whether retailer $i$ can benefits by doubling its replenishment interval depends only on $N[T_i;T_{-i}]$, the set of retailers with replenishment interval at most $T_i$.
% We say two joint policies $T^1$ and $T^2$ have the same order over $N$ if they satisfy that $T^1_i \leq T^1_j$ if and only if $T^2_i \leq T^2_j$ for any $i,j \in N$.
Then we obtain the following result.

\begin{lemma}[Better-response Monotonicity]
\label{lem:no-back-jump}
Let $T^1$ and $T^2$ be two joint policies such that $T^1_i=T^2_i$ for some retailer $i \in N$.
If $N[T^2_i;T^2_{-i}] \subseteq N[T^1_i;T^1_{-i}]$, then retailer $i$ can benefit by doubling its replenishment interval under $T^1$ only if it can benefit by doubling its replenishment interval under $T^2$, i.e., 
\[
f_i(2T^1_i;T^1_{-i}) < f_i(T^1_i;T^1_{-i})
\Rightarrow
f_i(2T^2_i;T^2_{-i}) < f_i(T^2_i;T^2_{-i}).
\]
In particular, if $N[T^1_i;T^1_{-i}]=N[T^2_i;T^2_{-i}]$, then
\[
f_i(2T^1_i;T^1_{-i}) < f_i(T^1_i;T^1_{-i})
\Leftrightarrow
f_i(2T^2_i;T^2_{-i}) < f_i(T^2_i;T^2_{-i}).
\]
\end{lemma}

% \begin{proof}
% According to Lemma \ref{lem:right-jump}, $f_i(2T^1_i;T^1_{-i}) < f_i(T^1_i;T^1_{-i})$ implies that
% \begin{align*}
%     T^2_i=T^1_i < \sqrt{\frac{1}{2}\frac{K_i+\frac{w_i}{\sum_{j \in N[T^1_i;T^1_{-i}]}w_j}K_0}{H_i}} 
%      \leq  \sqrt{\frac{1}{2}\frac{K_i+\frac{w_i}{\sum_{j \in N[T^2_i;T^2_{-i}]}w_j}K_0}{H_i}},
% \end{align*}
% which implies that $f_i(2T^2_i;T^2_{-i}) < f_i(T^2_i;T^2_{-i})$.
% \hfill$\square$ \end{proof}

Lemma \ref{lem:no-back-jump} implies the following monotonicity property of agents' better-response dynamics:
 % for retailers under our rule.
Suppose that under some joint policy~$T^1$, retailer $i$ doubles its replenishment interval from $T^1_i$ to $2T^1_i$, reducing its cost;
%and this change reduces its cost. 
After a sequence of changes by other retailers, we arrive at a new joint policy $T^2$;
Then retailer $i$ has no incentive to revert its replenishment interval from $2T^1_i$ back to $T^1_i$, provided that all retailers who initially had larger replenishment intervals  than $T^1_i$ under $T^1$ still have  replenishment intervals larger than $T^1_i$ under $T^2$.
In particular, the last condition can be satisfied if no retailer reduces its replenishment interval during this process.
This monotonicity property is crucial for analyzing the payoff dominant Nash equilibrium and its efficiency later in the paper.

Next we characterize the condition under which a retailer can reduce its costs by decreasing the replenishment interval.
The proof is similar to the proof of Lemma \ref{lem:right-jump}
and is deferred to Appendix~\ref{app:lem:left-jump}.

\begin{lemma}
\label{lem:left-jump}
For any \WPSshort rule and any joint policy $T$, the following statements are equivalent:
\begin{enumerate}
\item Retailer $i$ can reduce its cost by decreasing its replenishment interval to a certain extent;
\item Retailer $i$ can reduce its cost by halving its replenishment interval;
\item It holds that
    % \begin{equation}\label{eq:left-jump}
    \[
    T_i>\sqrt{2\frac{K_i+\frac{w_i}{\sum_{j \in N[T_i/2;T_{-i}]}w_j}K_0}{H_i}}.
    % \end{equation}
    \]
\end{enumerate}
\end{lemma}

% \begin{proof}
% The proof is analogous to the proof of Lemma \ref{lem:right-jump}.
% When retailer $i$ decreases its replenishment interval from $T_i$ to $T_i/2$, its holding cost is decreased by $H_iT_i/2$ while its setup cost is increased by 
% \[
% \frac{K_i+\frac{H_i}{\sum_{j \in N[T_i/2;T_{-i}]}H_j}K_0}{T^c_i}.
% \]
% This change reduces retailer $i$'s total cost if and only if
% \[
% H_iT_i/2>\frac{K_i+\frac{H_i}{\sum_{j \in N[T_i/2;T_{-i}]}H_j}K_0}{T_i}, 
% \]
% or
% \[
% T_i> \sqrt{2\frac{K_i+\frac{H_i}{\sum_{j \in N[T_i/2;T_{-i}]}H_j}K_0}{H_i}}.
% \]
% Thus statements (2) and (3) are equivalent. 

% Since (2) directly implies (1), it remains to show that (1) implies (2).
% Suppose that retailer $i$ can decrease its cost by decreasing its replenishment interval from $T_i$ to $2^{-z}T_i$ for some $z \in \mathbb{Z^+}$.
% After this change, the holding cost of retailer $i$ is decreased by $(1-2^{-z})H_iT_i$ while the setup cost is increased by at least
% \[
% \frac{K_i+\frac{H_i}{\sum_{j \in N[2^{-z}T_i;T_{-i}]}H_j}K_0}{2^{1-z}T_i}.
% \]
% Then
% \[
% (1-2^{-z})H_iT_i>\frac{K_i+\frac{H_i}{\sum_{j \in N[2^{-z}T_i;T_{-i}]}H_j}K_0}{2^{1-z}T_i}, 
% \]
% or
% \[
% T_i>\sqrt{\frac{2^{z-1}}{1-2^{-z}}\frac{K_i+\frac{H_i}{\sum_{j \in N[2^{-z}T_i;T_{-i}]}H_j}K_0}{H_i}} \geq \sqrt{2\frac{K_i+\frac{H_i}{\sum_{j \in N[T_i/2;T_{-i}]}H_j}K_0}{H_i}},
% \]
% which means that retailer $i$ can also reduce its cost by halving its replenishment interval.
% Thus, (1) implies (2).
% \hfill$\square$ \end{proof}

Based on Lemma \ref{lem:left-jump}, we show that if the replenishment intervals of all retailers are no smaller than that in the optimal centralized policy $T^c$, then any retailer whose replenishment interval equals that in $T^c$ does not benefit from decreasing it.

\begin{lemma}
\label{lem:no-left-jump}
For any \WPSshort rule and any joint policy $T$ with $T \geq T^c$, for any retailer $i$ with $T_i=T^c_i$, it cannot reduce its cost by decreasing its replenishment interval, i.e.,
$f_i(T^c_i;T_{-i}) \leq f_i(2^{-z}T^c_i;T_{-i}), \forall z\in \mathbb{Z^+}$.
% \[
% f_i(T^c_i;T_{-i}) \leq f_i(2^{-z}T^c_i;T_{-i}), \forall z\in \mathbb{Z^+}.
% \]
\end{lemma}

% \begin{proof}
% Suppose, for contradiction,  that retailer $i$ can reduce its cost by decreasing its replenishment interval from $T^c_i$ to $2^{-z}T^c_i$ for some $z \in \mathbb{Z^+}$.
% If $i \in U$, then $T^c_i  \leq \sqrt{2s}$, and $N[2^{-z}T^c_i;T_{-i}]=\{i\}$.
% According to Lemma \ref{lem:left-jump}, we have
% \[
% \sqrt{2s}\geq T^c_i >\sqrt{2\frac{K_i+K_0}{H_i}},
% \]
% which contradicts with the definition of $s$ in (\ref{eq:s-def}).
%  % $s= \min_{\emptyset \subset S \subseteq N}\frac{K_0+K(S)}{H(S)}$.
% If $i \in V$, then $T^c_i  \leq \sqrt{2\frac{K_i}{H_i}}$.
% According to Lemma \ref{lem:left-jump}, we have
% \[
% \sqrt{2\frac{K_i}{H_i}}\geq T^c_i >\sqrt{2\frac{K_i+\frac{w_i}{\sum_{j \in N[2^{-z}T^c_i;T_{-i}]}w_j}K_0}{H_i}},
% \]
% which is a contradiction since $w_i \geq 0$.
% Thus we get a contradiction for any retailer $i \in N$. 
% \hfill$\square$ \end{proof}

\subsection{Payoff Dominant Nash Equilibrium}
\label{sec:WPS-NE}

Based on the above properties, we establish the existence of a payoff dominant Nash equilibrium under all \WPSshort rules.
Intuitively, we will show that if we let all retailers start with the optimal centralized policy and let them decide for themselves whether they want to double, halve, or keep the replenishment interval, then they will reach a Nash equilibrium after a finite number of changes.
Moreover, this Nash equilibrium 
% is independent of the update orders of retailers and it 
is a payoff dominant Nash equilibrium.
We summarize our results in the following theorem.

% \begin{algorithm}
%   \caption{Compute the payoff dominant Nash equilibrium for the general proportionally sharing rule.}

%   Compute the optimal centralized policy~$T^c$
%   \While{condition}{
%     Perform some actions\;
%     \If{another condition}{
%       Do something else\;
%     }
%   }
  
%   \Return{Final result}\;
  
%   \label{alg:example}
% \end{algorithm}

\begin{theorem}
\label{thm:general-NE-exist}
% For \WPSshort rule with any $(w_i)_{i \in N}$, 
For any \WPSshort rule, 
there exists a payoff dominant Nash equilibrium, which can be found in $O\left(n^2\log_2\left(1+\frac{K_0}{\min_i K_i}\right)\right)$ time.
\end{theorem}

Theorem~\ref{thm:general-NE-exist} generalizes the result of \citet{he_noncooperative_2017} about the existence of a payoff dominant Nash equilibrium from the equal-division rule to all \WPSshort rules.
Moreover, to achieve this broader result, we conduct a deeper analysis of these rules by characterizing the jumping conditions and identifying an interesting monotonicity property of agent better responses (in Section~\ref{sec:WPS-property}).
In contrast, the result of \citet{he_noncooperative_2017} follows from the submodularity of the game under the equal-division rule.

We begin by presenting a simple algorithm (called Algorithm~1) for computing a Nash equilibrium under the \WPSshort rules. 
The algorithm initializes each retailer with the optimal centralized policy~$T^c$. 
Then, for each retailer $i \in N$, if it can reduce its cost by doubling or halving its replenishment interval, its interval is updated accordingly. 
This process is repeated until no retailer can further reduce its cost through such adjustments.
Notice that in Algorithm~1, we do not specify the update order of retailers in each round.
In fact, we will show that the update order does not affect the outcome of Algorithm~1.
% We start by showing that Algorithm~1 will stop in polynomial time with a Nash equilibrium.
We start by showing that no retailer halves its replenishment interval in Algorithm~1.

% \begin{algorithm}[t]
% \caption{Computing a Nash equilibrium under the \WPSshort rules}
% \label{alg:1}
% \begin{algorithmic}[1]
% % \State \textbf{Input:} Optimal centralized policy $T^c = (T_1^c, \dots, T_n^c)$
% \State \textbf{Initialize:} $T_i \gets T_i^c$ for all $i \in N$
% \Repeat
%     \State \textbf{For}  each retailer $i \in N$:
%     \State \hspace{1em} \textbf{If} $f_i(2T_i;\, T_{-i}) < f_i(T_i;\, T_{-i})$, \textbf{then} $T_i \gets 2T_i$
%     \State \hspace{1em} \textbf{Else if} $f_i(T_i / 2;\, T_{-i}) < f_i(T_i;\, T_{-i})$, \textbf{then} $T_i \gets T_i / 2$
% \Until{no retailer can reduce its cost by doubling or having $T_i$}
% % \State \textbf{Output:} Final vector $T = (T_1, \dots, T_n)$
% \end{algorithmic}
% \end{algorithm}

\begin{lemma}
\label{lem:no-left-jump-ALG}
During the process of Algorithm~1 no retailer decreases its replenishment interval.
\end{lemma}

\begin{proof}
We prove the result by induction.
At the beginning, according to Lemma \ref{lem:no-left-jump}, no retailer will halve its replenishment interval.
Suppose this holds for the first $t$ steps, and let us consider the next step.
For any retailer who has not changed its replenishment interval yet, it will not halve its replenishment interval according to Lemma \ref{lem:no-left-jump}.
For any remaining retailer $i$, by induction, its last change must be doubling its replenishment interval, which reduces its cost.
By induction, no retailer halves its replenishment interval up to now.
Then according to Lemma \ref{lem:no-back-jump}, halving retailer $i$'s replenishment interval will increase its cost.
Thus no retailer will halve its replenishment interval in the next step.
This finishes our induction.
Notice that the above proof does not depend on the update order of retailers in Algorithm~1.
\hfill$\square$ \end{proof}

Since retailers' strategy sets are bounded, it follows that Algorithm~1 will stop in polynomial time with a Nash equilibrium.

\begin{lemma}
\label{lem:alg1-NE}
Algorithm~1 finds a Nash equilibrium in $O\left(n^2\log_2\left(1+\frac{K_0}{\min_i K_i}\right)\right)$ time.
\end{lemma}

% \begin{proof}
% Since no retailer halves its replenishment interval (Lemma~\ref{lem:no-left-jump-ALG}) and retailer $i$'s strategy set is 
% \[
% \Gamma_i=\{T_i \mid \sqrt{K_i / (2H_i)} \le T_i \le \sqrt{2(K_0+K_i) / H_i} \text{ such that } T_i=2^{z_i}B, z_i \in \mathbb{Z}\},
% \]
% in Algorithm~1 retailer $i$ can make at most $\log_2 \sqrt{4\frac{K_0+K_i}{K_i}}$ number of changes.
% Overall, there are at most 
% % $n\log_2 \sqrt{4\left(1+\frac{K_0}{\min K_i}\right)}$ 
% $O\left(n\log_2\left(1+\frac{K_0}{\min_i K_i}\right)\right)$
% changes.
% In each round checking every $i \in N$, at least one retailer doubles its replenishment interval.
% So we get the running time $O\left(n^2\log_2\left(1+\frac{K_0}{\min_i K_i}\right)\right)$.
% Since no retailer wants to double or halve its replenishment interval at the final state, according to Lemma~\ref{lem:right-jump} and Lemma~\ref{lem:left-jump}, the final state is a Nash equilibrium. 
% \hfill$\square$ \end{proof}

Denote $T^w$ the Nash equilibrium found by Algorithm~1 under a \WPSshort rule with weight vector~$w$. 
We show that $T^w$ is the least Nash equilibrium, i.e., the replenishment interval of every retailer under $T^w$ is no larger than that under any other Nash equilibrium.

\begin{lemma}
\label{lem:alg1-least}
For any Nash equilibrium $T^*$, we have $T^w \leq T^*$.
\end{lemma}

\begin{proof}
% For any Nash equilibrium $T^*$, we show that $T^w \leq T^*$.
Denote $P=\{i \in N \mid T^w_i > T^*_i\}$ and suppose~$P \neq \emptyset$.
% the set of retailers whose replenishment interval in $T^w$ is strictly larger than that in $T^*$.
Let $P_1=\{i \in P \mid T^*_i \leq T^*_j, \forall j \in P\}$ be the set of retailers from $P$ with the smallest replenishment interval under $T^*$.
Further, let 
$P_2=\{i \in P_1 \mid T^w_i \geq T^w_j, \forall j \in P_1\}$
be the set of retailers from $P_1$ with the largest replenishment interval under $T^w$.
% In words, retailers in $P_1$ have the least replenishment interval in $T^*$ among all retailers from $P$ and retailers in $P_2$ have the largest replenishment interval in $T^w$ among all retailers from $P_1$.
Let $i_0 \in P_2$ be the retailer who first arrives at $T^w_{i_0}$ during Algorithm~1 (with an arbitrary update order) among all retailers in $P_2$.
Let us consider the step when retailer $i_0$ doubles its replenishment interval from $T^w_{i_0}/2$ to $T^w_{i_0}$.
According to Algorithm~1 and Lemma \ref{lem:right-jump}, we have
\begin{equation}
\label{eq:least-NE}
\frac{T^w_{i_0}}{2} < \sqrt{\frac{1}{2}\frac{K_{i_0}+\frac{w_{i_0}}{\sum_{j \in N_0}w_j}K_0}{H_{i_0}}},
\end{equation}
where $N_0$ is the set of retailers with replenishment interval at most $T^w_{i_0}/2$ just before $i_0$'s last doubling.

By the choice of $i_0$ we have $P_2 \subseteq N_0$.
By the definition of $P_1$ and $P_2$, we have $T^w_i\leq T^w_{i_0}/2$ for any $i \in P_1 \setminus P_2$.
Since no retailer decreases its replenishment interval during Algorithm~1, we have $P_1 \setminus P_2 \subseteq N_0$.
In addition, denote $Q=\{i \in N \setminus P \mid T^*_i \leq T^*_{i_0}\}$.
For any $i \in Q$, since $i \notin P$ and $i_0 \in P$, we have 
% $T^w_i \leq T^*_i \leq T^*_{i_0}<T^w_{i_0}$.
\[
T^w_i \overset{i \notin P}{\leq} T^*_i \overset{i \in Q}{\leq} T^*_{i_0} \overset{i_0 \in P}{<} T^w_{i_0} \Rightarrow T^w_i \leq T^w_{i_0}/2.
\]
% That is, $T^w_i \leq T^w_{i_0}/2$.
Again, since no retailer decreases its replenishment interval during Algorithm~1, it follows that $Q \subseteq N_0$.
Combining $P_2\subseteq N_0$, $P_1 \setminus P_2\subseteq N_0$, and $Q \subseteq N_0$, we get $P_1 \cup Q \subseteq N_0$.

By the definition of $P_1$ and $i_0 \in P_1$, we have $P_1=\{i \in P \mid T^*_i \leq T^*_{i_0}\}$.
Now for retailer $i_0$ under $T^*$ we have $N[T^*_{i_0};T^*_{-i_0}]=P_1 \cup Q \subseteq N_0$.
Then
\begin{align*}
    T^*_{i_0}  \overset{i_0 \in P}{\leq} \frac{T^w_{i_0}}{2} \overset{(\ref{eq:least-NE})}{<} \sqrt{\frac{1}{2}\frac{K_{i_0}+\frac{w_{i_0}}{\sum_{j \in N_0}w_j}K_0}{H_{i_0}}} 
% \overset{N[T^*_{i_0};T^*_{-i_0}] \subseteq N_0}{\leq} 
 \leq
\sqrt{\frac{1}{2}\frac{K_{i_0}+\frac{w_{i_0}}{\sum_{j \in N[T^*_{i_0};T^*_{-i_0}]}w_j}K_0}{H_{i_0}}}.
\end{align*}

According to Lemma \ref{lem:right-jump}, retailer $i_0$ can reduce its cost by doubling its replenishment interval, which contradicts with that $T^*$ is a Nash equilibrium.
Therefore, $P=\emptyset$ and hence $T^w \leq T^*$.
\hfill$\square$ \end{proof}

Since the proof of Lemma~\ref{lem:alg1-least} does not depend on the update order of retailers in Algorithm~1, we can conclude that Algorithm~1 always outputs the unique least Nash equilibrium $T^w$, regardless of the update order.
Finally, we show that the least Nash equilibrium $T^w$ is payoff dominant.
 
\begin{lemma}
\label{lem:alg1-dominant}
$T^w$ is a payoff dominant Nash equilibrium.
\end{lemma}

\begin{proof}
Let $T^*$ be an arbitrary Nash equilibrium with $T^* \neq T^w$.
For any retailer $i \in N$, we show that its cost under $T^*$ is no smaller than that under $T^w$.
Since $T^w$ is a Nash equilibrium, we have 
\[
f_i(T^w_i,T^w_{-i}) \leq f_i(T^*_i,T^w_{-i}).
\]
Next we compare $f_i(T^*_i,T^w_{-i})$ and $f_i(T^*_i,T^*_{-i})$.
According to Lemma \ref{lem:alg1-least}, we have $T^* \geq T^w$.
This implies that, for each joint order including retailer $i$, the set of other retailers included in this order under the policy $(T^*_i,T^*_{-i})$ is a subset of that under the policy $(T^*_i,T^w_{-i})$.
Thus, for retailer $i$, the average sharing for the major setup cost under $(T^*_i,T^w_{-i})$ is no greater than that under $(T^*_i,T^*_{-i})$.
Since the holding cost and the minor setup cost for retailer $i$ remain the same, we can conclude that 
\[
f_i(T^w_i,T^w_{-i}) \leq f_i(T^*_i,T^w_{-i}) \leq f_i(T^*_i,T^*_{-i}).
\]
Therefore, $T^w$ is a payoff dominant Nash equilibrium.
\hfill$\square$ \end{proof}

By combining Lemma \ref{lem:alg1-NE} and Lemma \ref{lem:alg1-dominant}, we obtain the result stated in
Theorem~\ref{thm:general-NE-exist}.

% \subsection{Efficiency of Nash Equilibrium}
\subsection{PoS as a Measure of Efficiency}
\label{sec:WPS-efficiency}

Having established the existence of payoff dominant Nash equilibrium, the next step is to evaluate and compare different \WPSshort rules.
We start by showing that all reasonable cost allocation rules have essentially the same large PoA values.
\citet{he_noncooperative_2017} show that the PoA of the equal-division rule can be bounded from above by~$\frac{3}{2\sqrt{2}}\sqrt{n}$.
This result can be easily generalized for an arbitrary allocation rule. 
% The proof of Proposition~\ref{prop:POA-bound-upper} is deferred to Appendix~\ref{app:prop:POA-bound-upper}.
% satisfying property P2.
% In fact, the proof of Theorem 4 of \citet{he_noncooperative_2017} does not need any condition of the cost allocation rule, this bound holds for any cost allocation rule.

% To this end, we propose the following two basic properties.
%  % that any reasonable rule should satisfy.
% \begin{enumerate}
%   \item[P1] Equal treatment of equals: For any joint replenishment policy $T$, if there are two retailers $i,j \in N$ with $K_i=K_j$, $h_i=h_j$, $d_i=d_j$, and $T_i=T_j$, then $f_i(T)=f_j(T)$.
%   % \item[P2] No sharing for minor setup costs or holding costs: $f_i(T_i;T_{-i}) \geq H_iT_i+\frac{K_i}{T_i}$ for any $i \in N$ under any $T$.
%   \item[P2] Major setup cost of every joint order is shared only by retailers who placed this joint order.
% \end{enumerate}
% Notice that both properties are satisfied by all \WPSshort rules, so the set of rules satisfying the above two properties is even broader than the set of all \WPSshort rules. 

\begin{proposition}
\label{prop:POA-bound-upper}
The price of anarchy of any cost allocation rule is at most $\frac{3}{2\sqrt{2}}\sqrt{n}$.
\end{proposition}

% \begin{proof}
% By property P2, the share of the major setup cost for any retailer $i$ is at most $\frac{K_0}{T_i}$ under any joint replenishment policy $T$, i.e., 
% \[
% x_i(T_i;T_{-i}) \leq \frac{K_0}{T_i}, \forall i \in N, \forall T \in \Gamma.
% \]
% Let $T^*$ be an arbitrary Nash equilibrium. We have
% \begin{align*}
% c(T^*) &=\sum_{i=1}^n f_i(T_i^*;T_{-i}^*) \leq \sum_{i=1}^n \min_{T_i \in \Gamma_i} f_i(T_i;T_{-i}^*) \\
% &= \sum_{i=1}^n \min_{T_i \in \Gamma_i} \left(H_iT_i+\frac{K_i}{T_i}+x_i(T_i^*;T_{-i}^*)\right) \\
% &\leq \sum_{i=1}^n \min_{T_i \in \Gamma_i} \left(H_iT_i+\frac{K_i+K_0}{T_i}\right). 
% \end{align*}
% On the other hand, for the social optimum we have
% \begin{align*}
% c(T^c) &=\frac{K_0}{T_{\min}^c}+\sum_{i=1}^n \left(H_iT_i^c+\frac{K_i}{T_i^c}\right) \\
% &\geq \sum_{i=1}^n \left(H_iT_i^c+\frac{K_i+K_0/n}{T_i^c}\right) \\
% &\geq \sum_{i=1}^n \min_{T_i \in \Gamma_i} \left(H_iT_i+\frac{K_i+K_0/n}{T_i}\right). 
% % &\geq \sum_{i=1}^n \frac{3}{2\sqrt{2}} \cdot 2 \cdot \sqrt{H_i(K_i+K_0/n)} 
% \end{align*}
% So the PoA is at most 
% \[\frac{c(T^*)}{c(T^c)} \leq \frac{\sum_{i=1}^n \min_{T_i \in \Gamma_i} \left(H_iT_i+\frac{K_i+K_0}{T_i}\right)}{\sum_{i=1}^n \min_{T_i \in \Gamma_i} \left(H_iT_i+\frac{K_i+K_0/n}{T_i}\right)}
% \leq \frac{3}{2\sqrt{2}}\sqrt{n},
% \]
% where the last inequality follows from the fact that $2\sqrt{ab} \leq \min_{y\in\{2^z \mid z \in \mathbb{Z}\}} ay+\frac{b}{y} \leq \frac{3}{\sqrt{2}}\sqrt{ab}$ for any $a,b \geq 0$.
% \hfill$\square$ \end{proof}

We complement this result by providing a matching lower bound (up to a constant) for cost allocation rules that satisfy symmetry.
A cost allocation rule satisfies \emph{symmetry} if, for any joint replenishment policy $T$, whenever two retailers $i,j \in N$ have identical parameters, i.e., $K_i=K_j$, $h_i=h_j$, $d_i=d_j$, and $T_i=T_j$, their allocated costs are equal, i.e., $f_i(T)=f_j(T)$.
Notice that all \WPSshort rules satisfy symmetry.

\begin{proposition}
\label{prop:POA-bound}
The price of anarchy of any cost allocation rule satisfying symmetry (including all \WPSshort rules) is at least $\frac{\sqrt{2}}{3}\sqrt{n}$.
\end{proposition}

% \begin{proof}
% Let $K_0=1$. Let $K_i=0$, $h_i=1$, $d_i=2$, and $H_i=\frac{1}{2}h_id_i=1$ for each retailer $i \in N$. 
% For the optimal centralized policy $T^c$ we have $s=\frac{1}{n}$ and $C(T^c) \leq (\sqrt{2}+\frac{1}{\sqrt{2}})\sqrt{n}=\frac{3}{\sqrt{2}}\sqrt{n}$.
% On the other hand, for any rule satisfying symmetry, there is a Nash equilibrium $\tilde{T}$ with $\tilde{T}_i=1$.
% Indeed, the cost of every retailer under $\tilde{T}$ is $H_i\tilde{T}_i+\frac{1}{n}\frac{K_0}{\tilde{T}_i}=1+\frac{1}{n}$ by symmetry.
% If a retailer doubles its replenishment interval, its cost will be at least $2H_i\tilde{T}_i=2>1+\frac{1}{n}$.
% If a retailer halves its replenishment interval, its cost will be at least $\frac{1}{2}H_i\tilde{T}_i+\frac{1}{2}\frac{K_0}{\tilde{T}_i/2}=1.5 \geq 1+\frac{1}{n}$.
%  % by property P2.
% Since $C(\tilde{T})=n+1$, the PoA is at least $\frac{n+1}{\frac{3}{\sqrt{2}}\sqrt{n}} \geq \frac{\sqrt{2}}{3}\sqrt{n}$.
% \hfill$\square$ \end{proof}

Combining Proposition \ref{prop:POA-bound-upper} and Proposition \ref{prop:POA-bound}, we can conclude that all rules satisfying symmetry have PoAs in $[\frac{\sqrt{2}}{3}\sqrt{n},\frac{3}{2\sqrt{2}}\sqrt{n}]$.
In other words, all reasonable rules will have large PoA values that differ by at most a constant.
Therefore, PoA is not a good metric to compare the efficiency of different cost allocation rules.
In contrast, we will show that PoS can differentiate the efficiency of various \WPSshort rules.
Moreover, since for \WPSshort rules a payoff dominant Nash equilibrium always exists and can be efficiently found in polynomial time (Theorem \ref{thm:general-NE-exist}), it is particularly appealing to consider PoS, which precisely quantifies the efficiency loss of the payoff dominant Nash equilibrium. 
Indeed, the payoff dominance, introduced by \citet{harsanyi1988general}, has been widely employed for selecting among Nash equilibria in non-cooperative games~\citep{colman1997payoff}.
Therefore, we will use PoS to compare the performance of cost allocation rules. 

\section{Developing Efficient and Practical \WPSshort Rules}
\label{sec:efficient-rules}

In this section, we focus on identifying practical \WPSshort rules that can guarantee strong coordination efficiency, as measured by PoS.
We begin by presenting a theoretical benchmark rule that achieves the optimal PoS of 1 but suffers from significant practical limitations (Section~\ref{sec:WPS-PoS-1}).
We then propose a simple and practical rule that overcomes these limitations while achieving near-optimal efficiency (Section~\ref{sec:WPSh}).

\subsection{\WPSo: Theoretically Optimal Benchmark}
\label{sec:WPS-PoS-1}

We first present a \WPSshort rule, called \WPSo, that achieves the theoretically optimal PoS of 1.
This rule demonstrates that, in principle, full coordination without efficiency loss is possible in a decentralized setting.
The weight vector $w^*$ of \WPSo is defined as $w^*_i = \frac{H_is-K_i}{K_0}$ if $i \in U$ and $w^*_i = 0$ if $i \in V$,
% \begin{align}
% \label{optimal-w}
% w^*_i = 
% \begin{cases}
%     \frac{H_is-K_i}{K_0} & \text{if } i \in U; \\
%     0 & \text{if } i \in V,
% \end{cases}
% \end{align}
where the sets $U$ and $V$ correspond to the retailer partition under the optimal centralized policy~$T^c$.
Recall that under~$T^c$, retailers in $U$ order together with the minimum replenishment interval $T^c_{\min}$, while those in $V$ each follow their individual EOQ policy, ignoring the major setup cost.
The weight vector of \WPSo is carefully constructed to ensure that the centralized solution $T^c$ forms a Nash equilibrium.
% To achieve the optimal PoS of 1, we need to ensure that $T^c$ is a Nash equilibrium under \WPSo.
Indeed, for each retailer in~$V$, since their weight is zero, they do not share the major setup cost and thus have no incentive to deviate from their EOQ interval.
For retailers in~$U$, the weights ensure that the total cost (including both setup costs and holding costs) is shared proportionally to their $H_i$ values, i.e., 
$
    f_i(T^c)=
\frac{H_i}{\sum_{j \in U}H_j}\left(\sum_{j \in U}H_jT^c_{\min}+\frac{K_0+\sum_{j \in U}K_j}{T^c_{\min}}\right)
$.
Then, for each retailer from $U$, the ratio between its holding cost and setup cost matches the ratio for the whole group $U$, which is optimized by $T^c_{\min}$.
Consequently, we can show that doubling its replenishment interval is not beneficial.
Therefore, the optimal centralized policy $T^c$ is a Nash equilibrium under \WPSo, implying that the PoS is 1.
% The proof of Proposition~\ref{prop:k-know-POS=1} is deferred to Appendix~\ref{app:prop:k-know-POS=1}.

\begin{proposition}
\label{prop:k-know-POS=1}
% When both $K_i$ and $h_i$ are public information, 
The price of stability of  \WPSo is~1.
% \begin{align}
% \label{optimal-w}
% w^*_i = 
% \begin{cases}
%     \frac{H_is-K_i}{K_0} & \text{if } i \in U; \\
%     0 & \text{if } i \in V.
% \end{cases}
% \end{align}
% \[
% w_i
% % =\frac{x_i(T^c_i,T^c_{-i})}{\frac{K_0}{T^c_{\min}}}\frac{1}{\sum_{j \in U}H_jK_0}((K_0+\sum_{j \in U}K_j)H_i-\sum_{j \in U}H_jK_i))
% =\frac{H_is^2-K_i}{K_0}.
% \]
\end{proposition}

We remark that all retailers' costs under \WPSo are the same as the core allocation  described by \citet{anily_cost_2007}, who study a cooperative game-theoretical joint replenishment model. % and show that the allocation under this rule is a core allocation.
While the core concept in cooperative games and the Nash equilibrium concept in non-cooperative games are fundamentally different, Proposition~\ref{prop:k-know-POS=1} shows that, in our non-cooperative game model, where each retailer can freely decides its own replenishment interval, \WPSo ensures that the optimal centralized policy remains stable.

Although \WPSo achieves the optimal PoS of 1, it suffers from two key drawbacks that hinder its practical applicability in decentralized settings.
First, all retailers in $V$ are assigned zero weight and hence do not contribute to the major setup cost.
These retailers effectively act as free-riders, which may create fairness concerns among contributing retailers in $U$ and ultimately threaten the stability of the joint system.\footnote{This free-rider issue was also noted as a major drawback in the cooperative setting by~\citet{anily_cost_2007}.}
Second, the weight assigned to each retailer in $U$ depends not only on its own parameters $H_i$ and $K_i$, but also on the full set of parameters from all other retailers of $U$.
Specifically, recall that $s=\frac{K_0+\sum_{j \in U}K_j}{\sum_{j \in U}H_j} = \min_{\emptyset \subset S \subseteq N}\frac{K_0+K(S)}{H(S)}$, which requires global knowledge of the entire system.
This level of informational dependency is impractical in decentralized settings.

% \section{Proportional Sharing by $H_i$}
\subsection{\WPSh: Simple and Near-Optimal Solution}
\label{sec:WPSh}

% In this section, we assume that the information about $H_i$ is known to the central designer.
% For this setting, 
To overcome the practical limitations of \WPSo, we propose a simple and practical alternative, denoted \WPSh.
This rule belongs to the class of \WPSshort rules and sets the weight of each retailer $i$ as $H_i=\frac{1}{2}h_id_i$.
Under \WPSh, whenever a subset $S \subseteq N$ of retailers places a joint order, the corresponding major setup cost is shared by retailers in $S$ proportionally to their respective $H_i$ values. 
Note that \WPSh avoids the issue of free-riders, as every participating retailer contributes to the joint cost.
Furthermore, the weight assigned to each retailer depends only on its own parameter $H_i$, without requiring information about other retailers.
These properties make \WPSh both simple and well-suited for decentralized scenarios.

Since \WPSh is one of the \WPSshort rules, the existence of a payoff dominant Nash equilibrium is already guaranteed by Theorem \ref{thm:general-NE-exist}.
% In addition, for \WPSh we have $\gamma=1$, so its PoS is at most a constant number according to Theorem~\ref{}.
In the remainder of this subsection, 
% we provide a more refined analysis for \WPSh.
we provide a linear-time algorithm to find the payoff dominant Nash equilibrium under \WPSh and show that its PoS is upper bounded by~1.25.

% For a joint replenishment policy $T$, 
% denote $N[T_i;T_{-i}]$ the set of retailers with replenishment interval at most $T_i$ under the joint policy $(T_i,T_{-i})$.
% The long-run average cost of retailer $i$ is
% \[
% f_i(T_i;T_{-i})=H_iT_i+\frac{K_i}{T_i}+x_i(T_i;T_{-i}),
% \]
% where 
% \[
% x_i(T_i;T_{-i})=\sum_{m=i}^{n}(\frac{1}{T_m}-\frac{1}{T_{m+1}})\frac{H_i}{\sum_{j \in N[T_m;T_{-m}]}H_j}K_0.
% \]
% Denote 
% \[
% x_i^+(T_i;T_{-i})=x_i(2T_i;T_{-i})-x_i(T_i;T_{-i})=\frac{H_i}{\sum_{j \in N[T_i;T_{-i}]}H_j}\frac{K_0}{2T_i},
% \]
% and 
% \[
% x_i^-(T_i/2;T_{-i})=x_i(T_i;T_{-i})-x_i(T_i/2;T_{-i})=\frac{H_i}{\sum_{j \in N[T_i/2;T_{-i}]}H_j}\frac{K_0}{T_i}.
% \]

\subsubsection{Linear-Time Algorithm}

% We provide a linear-time algorithm to compute a payoff dominant Nash equilibrium under \WPSh.
 % described in the above proof.
The idea is that if we let retailers make decisions in the reverse order of their $\frac{K_i}{H_i}$ in Algorithm 1, then we can quickly reach a Nash equilibrium after a subset of retailers doubling their replenishment intervals.
Our algorithm (called Algorithm~2) works as follows.
Assume that the retailers are ordered such that $\frac{K_1}{H_1} \leq \frac{K_2}{H_2} \leq \dots \leq \frac{K_n}{H_n}$.
Initially, all retailers follow the optimal centralized policy $T^c$.
Then, according to the reverse order of the $K_i/H_i$ ratios (i.e., from retailer $n$ to retailer $1$), the algorithm examines each retailer $i$ to determine whether doubling its replenishment interval reduces its cost, according to Lemma \ref{lem:right-jump}.
If so, the replenishment interval of retailer~$i$ is doubled.
The algorithm ends after completing one full pass over all retailers.
Let $T^h$ denote the joint policy returned by Algorithm~2. 
For each $i \in \{1, 2, \dots, n\}$, let $T^i$ denote the joint policy after retailer $i$'s replenishment interval has been updated (if applicable) during the execution of the algorithm. 
Then the sequence of generated joint policies throughout the execution of the algorithm is thus: $T^c \rightarrow T^n \rightarrow T^{n-1} \rightarrow \dots \rightarrow T^1=T^h$.

% \begin{algorithm}[t]
% \caption{Computing a Nash equilibrium under  \WPSh}
% \label{alg:2}
% \begin{algorithmic}[1]
% % \State \textbf{Input:} Centralized optimal policy $T^c = (T_1^c, \dots, T_n^c)$; parameters $K_i$, $H_i$
% \State \textbf{Assume:} Retailers are ordered so that $\frac{K_1}{H_1} \leq \dots \leq \frac{K_n}{H_n}$
% \State \textbf{Initialize:} $T_i \gets T_i^c$ for all $i \in N$
% \State \textbf{For} $i = n$ to $1$:
% \State \hspace{1em} \textbf{If} $f_i(2T_i;\, T_{-i}) < f_i(T_i;\, T_{-i})$, \textbf{then} $T_i \gets 2T_i$
% % \State \textbf{Output:} Policy $T^h = (T_1, \dots, T_n)$
% \end{algorithmic}
% \end{algorithm}

The reason behind that Algorithm~2 traverses retailers in the descending order of their $K_i/H_i$ values is as follows.
According to Lemma \ref{lem:right-jump}, the condition for a beneficial doubling under a joint policy $T$ and \WPSh with weight $w_i=H_i$ is 
\begin{align*}
    T_i  <\sqrt{\frac{1}{2}\frac{K_i+\frac{w_i}{\sum_{j \in N[T_i;T_{-i}]}w_j}K_0}{H_i}} 
    = \sqrt{\frac{1}{2}\left(\frac{K_i}{H_i}+\frac{K_0}{\sum_{j \in N[T_i;T_{-i}]}w_j}\right)},
\end{align*}
which is more strict for agents with smaller $K_i/H_i$.
Thus, if retailer $i$ does not double its replenishment interval, then neither does any retailer who has a smaller $K_i/H_i$ and the same $T_i$, leading to the following.
% This is formally stated in the following lemma and the proof is deferred to Appendix~\ref{app:lem:no-double}.

\begin{lemma}
\label{lem:no-double}
For any $i \in N$, if $T^h_i=T^c_i$, then $T^h_j=T^c_j$ for every retailer $j$ with $j \leq i$ and $T^c_j=T^c_i$.
\end{lemma}

Since Algorithm~2 traverses each retailer only once, the doubling happens at most once for each retailer and we have $T^h_i \in\{T^c_i, 2 T^c_i\} $ for each $i \in N$.
Thus, for any two retailers $i,j \in N$ with $T^c_j<T^c_i$, it holds that $T^h_j \le T^h_i$.
Together with Lemma \ref{lem:no-double}, we conclude that for any retailer $i$ who does not double its replenishment interval during Algorithm~2, the set of retailers whose replenishment intervals are no larger than retailer $i$'s under $T^{i}$ remains the same under $T^h$ (at the end of Algorithm~2).
% The proof of Lemma~\ref{lem:same-N} is deferred to Appendix~\ref{app:lem:same-N}.

\begin{lemma}
\label{lem:same-N}
For any retailer $i \in N$ with $T^h_i=T^c_i$, we have $N[T^i_i;T^i_{-i}]=N[T^h_i;T^h_{-i}]$.
\end{lemma}

% \begin{proof}
% We partition the retailer set $N$ into four subsets:
% $S^1=\{j \in N \mid j \leq i \text{ and } T^i_j<T^i_i\}$, $S^2=\{j \in N \mid j \leq i \text{ and } T^i_j=T^i_i\}$, $S^3=\{j \in N \mid j > i \text{ and } T^i_j=T^i_i\}$, and $S^4=\{j \in N \mid j > i \text{ and } T^i_j> T^i_i\}$.
% By definition, $N[T^i_i;T^i_{-i}]=S^1 \cup S^2 \cup S^3$.
% Since the replenishment intervals of retailers from $S^3 \cup S^4$ do not change from $T^i$ to $T^h$, we have $S^3 \subseteq N[T^h_i;T^h_{-i}]$ and $S^4 \cap N[T^h_i;T^h_{-i}]=\emptyset$.
% Then it suffices to show $S^1 \cup S^2 \subseteq N[T^h_i;T^h_{-i}]$, or equivalently, $T^h_j \leq T^h_i$ for any $j \in S^1 \cup S^2$. 
% % Denote $P= \{j \in N \mid j \leq i \text{ and } T^c_j=T^c_i\}\subseteq S^1$.
% Since $T^h_i=T^c_i$, we have $T^h_i=T^c_i=T^i_i$.
% According to Lemma \ref{lem:no-double} and $T^h_i=T^c_i$, we have $T^h_j=T^i_j=T^i_i=T^h_i$ for any $j \in S^2$.
% For any $j \in S^1$, from $T^i_j<T^i_i$ we have $T^i_j \leq \frac{1}{2}T^i_i= \frac{1}{2}T^h_i$.
% Since every retailer doubles its replenishment interval at most once during Algorithm~2, we have $T^h_j \leq 2T^i_j \leq T^h_i$ for any $j \in S^1$.
% \hfill$\square$ \end{proof}

Now we show that $T^h$ is indeed a payoff dominant Nash equilibrium.

\begin{theorem}
\label{thm:NE-alg}
Under \WPSh, Algorithm~2 finds a payoff dominant Nash equilibrium in $O(n)$ time.
\end{theorem}

\subsubsection{Near-Optimal Price of Stability}

% According to Theorem \ref{}, we know that PoS under \WPSh is upper bounded by a constant since $\gamma=1$ for \WPSh.
% In this subsection, we give a more precise analysis of the upper bound.
We proceed to analyze the PoS of \WPSh.
Since $T^h_i \le 2T^c_i$ for every $i \in N$, it immediately follows that the PoS is at most 2.
To get a tighter upper bound, we explore the influence POT policies.
We show that if a retailer's replenishment interval is doubled during Algorithm~2, then it must be that the interval under $T^c$ was rounded down to the nearest POT point.
% The proof of Lemma~\ref{lem:POT-left} is deferred to Appendix~\ref{app:lem:POT-left}.

\begin{lemma}
\label{lem:POT-left}
For any retailer $i \in U$, if $T^h_i=2T^c_i$, then $T^c_{\min}=T^c_i=\alpha \sqrt{s}$ for some $\alpha \in [\frac{1}{\sqrt{2}},1)$.
For any retailer $i \in V$, if $T^h_i=2T^c_i$, then $T^c_i=\alpha_i \sqrt{\frac{K_i}{H_i}}$ for some $\alpha_i \in [\frac{1}{\sqrt{2}},1)$.
\end{lemma}

Using Lemma~\ref{lem:POT-left}, we prove a tighter upper bound for the PoS under \WPSh.

\begin{theorem}
\label{thm:POS}
Under \WPSh, the price of stability is bounded above by $1.25$.
\end{theorem}

Theorem~\ref{thm:NE-alg} implies that we can achieve a near-optimal PoS without knowing $K_i$.
We complement this result by showing that the PoS of any \WPSshort rule is strictly larger than 1 when $K_i$ is private information.
% In other words, the optimal centralized policy $T^c$ cannot be a Nash equilibrium.

\begin{proposition}
\label{prop:k-unknow-POS>1}
When $K_i$ is private information, the PoS of any \WPSshort rule is at least~1.05.
\end{proposition}

% The proof of Proposition \ref{prop:k-unknow-POS>1} is deferred to Appendix~\ref{app:prop:k-unknow-POS>1}.
% To prove this lower bound, we construct two instance profiles with two retailers, both having $K_0 = 5$ and $H_1 = H_2 = 1$.
% We set $K_1=1$ and $K_2=6$ for the first instance, and $K'_1=6$ and $K'_2=1$ for the second.
% Since the two instances differ only in $K_i$, which is assumed to be private information, the weight vector~$(w_1,w_2)$ must remain the same for them.
% Assume without loss of generality that $w_1 \leq w_2$.
% Then for the first instance the payoff dominant Nash equilibrium is $T^w=(2,4)$ with total cost of $C(T^w)=10.5$, whereas the optimal centralized policy is $T^c=(2,2)$ with the total cost of $C(T^c)=10$.

% \begin{proposition}
% \label{prop:k=0}
% When $K_i=0$ for every $i \in N$, the price of stability is 1 under \WPSh.
% \end{proposition}

% \begin{proof}
% We show that $T^c$ is a Nash equilibrium.
% Since $K_i=0$ for every $i \in N$, we have $U=N$ and $T^c_i=T^c_{\min}$ for every $i \in N$.
% For each retailer $i \in N$, we have 
% \[
% T^c_{i}=T^c_{\min} \overset{POT}{=} \sqrt{\frac{K_0}{\sum_{j \in N}H_j}} \geq \sqrt{\frac{1}{2}\frac{K_0}{\sum_{j \in N}H_j}}.
% \]
% According to Lemma \ref{lem:right-jump}, retailer $i$ cannot reduce its cost by increasing its replenishment interval.
% Combining with Lemma \ref{lem:no-left-jump}, we get that $T^c$ is a Nash equilibrium, and hence the price of stability is 1.
% \hfill$\square$ \end{proof}

% \section{Proportional Sharing by $d_i$} 
\section{Coordination Under Information Constraints} 
\label{sec:H-unkown}

While \WPSh is simple and reasonable, with its PoS close to 1, it can still be unsatisfactory  when the exact information $H_i=\frac{1}{2}h_id_i$ for each retailer is unavailable.
% , which could be retailers' private information in some scenarios.
In this section, we consider the scenario where both $H_i$ and $K_i$ are private information.
We first establish a general upper bound on the PoS that applies to all \WPSshort rules (Section~\ref{sec:PoS-general-bound}).
Then we propose a new rule that does not require any private information from retailers while still achieving strong efficiency in practice (Section~\ref{sec:WPSd}), and another rule based on estimated values of $H_i$ (Section~\ref{sec:estimation}).

\subsection{PoS Upper Bound for General \WPSshort Rules}
\label{sec:PoS-general-bound} 

We begin by analyzing the efficiency of arbitrary \WPSshort rules and derive general upper bounds on their PoS.
Let \WPSw denote a \WPSshort rule with weight vector~$w=(w_i)_{i \in N}$.
We will show that the PoS of \WPSw crucially depends on the ratio $\gamma_w \coloneqq \frac{\max_i H_i/w_i}{\min_i H_i/w_i}$.
Our analysis begins with a weaker upper bound on the PoS obtained by comparing each retailer’s replenishment interval in the payoff-dominant Nash equilibrium $T^w$ to its optimal interval in the centralized policy~$T^c$.
% The proof of Proposition~\ref{prop:jump-bound} is deferred to Appendix~\ref{app:prop:jump-bound}.

\begin{proposition}
\label{prop:jump-bound}
For \WPSw, we have $\max_{i \in N}\frac{T^w_i}{T^c_i} \leq 2\sqrt{1+\gamma_w}$ and hence the price of stability is bounded above by $2\sqrt{1+\gamma_w}$.
% \[
% \max_{i \in N}\frac{T^w_i}{T^c_i} \leq 2\sqrt{1+\gamma_w}.
% \]
\end{proposition}

Next, we introduce a refined analytical approach that yields a significantly stronger PoS bound.
Instead of analyzing retailers individually, we group them based on their ratios $H_i/w_i$ and collectively bound the cost for retailers in the same group.
In particular, for retailers in $U$, we partition them into $m \coloneqq \lceil \log \gamma_w \rceil$ groups such that the $H_i/w_i$ values within each group differ by at most a factor of 2.
We then bound the holding cost of each group using the largest replenishment interval in that group.
Based on this partition, we can prove that the total holding cost for retailers in $U$ is at most $\left(\sqrt{m}+\frac{1}{\sqrt{2}}\right)C(T^c)$. 
For retailers in $V$, in addition to $H_i/w_i$, we need to further partition retailers with similar $H_i/w_i$ based on their $K_i/H_i$ and employ a more intricate technique to bound the total holding cost by $3\sqrt{m}C(T^c)$.
Since the total setup cost under $T^w$ does not exceed that under $T^c$, we arrive at the following bound.
% The proof of Theorem~\ref{thm:general-POS} is deferred to Appendix~\ref{app:thm:general-POS}.

\begin{theorem}
\label{thm:general-POS}
For \WPSw, the price of stability is bounded above by $4\sqrt{\lceil \log \gamma_w \rceil}+\frac{1}{\sqrt{2}}+1$.
% where $\gamma_w \coloneqq \frac{\max_i H_i/w_i}{\min_i H_i/w_i}$.
\end{theorem}

\subsection{\WPSd: Efficient Rule without Private Information}
\label{sec:WPSd}

For the scenario where both $H_i$ and $K_i$ are private information, 
we propose a new rule, called \WPSd, under which each retailer $i$'s weight is exactly its own demand rate~$d_i$.
That is, whenever a subset $S \subseteq N$ of retailers place a joint order, the corresponding major setup cost is shared by retailers in $S$ proportionally to their demand rates $d_i$.
% \begin{itemize}
%   \item Every retailer pays its own minor setup cost and holding cost;
%   \item When a joint order is placed, all ordering retailers share the major setup cost proportionally to their demand rates~$d_i$.
% \end{itemize} 
Importantly, each retailer’s demand rate $d_i$ can be directly computed according to its order quantity $d_iT_i$ and replenishment interval $T_i$, both of which are observable and verifiable in practice.
Therefore, demand rates $d_i$ can be treated as public information, and \WPSd does not need any private cost information from retailers.

Using Theorem~\ref{thm:general-POS}, we derive the following upper bound for the PoS of \WPSd.

\begin{corollary}
\label{cor:WPSd-pos}
For \WPSd, the price of stability is bounded above by $4\sqrt{\lceil \log \gamma_d \rceil}+\frac{1}{\sqrt{2}}+1$, where
$\gamma_d \coloneqq \frac{\max_i H_i/d_i}{\min_i H_i/d_i}=\frac{\max_i h_i}{\min_i h_i}$.
\end{corollary}

In real-world applications, the variation in the holding cost rate~$h_i$ among retailers is typically not substantial, so the value of $\sqrt{\lceil \log \gamma_d \rceil}$ is usually very small.
Thus \WPSd can guarantee a good performance for the decentralized system in practice while requiring no private information.
In contrast, the PoS of the equal-division rule studied in \citet{he_noncooperative_2017} is shown to be $O(\sqrt{\ln n})$.
We note that \WPSd is similar in spirit to another rule analyzed in \citet{he_noncooperative_2017}, called Proportional Sharing Rule (PSR), where retailers share the major setup cost proportionally to their order quantity~$d_iT_i$.
However, as shown by \citet{he_noncooperative_2017}, a key drawback of PSR is the lack of a guaranteed Nash equilibrium.
Our results demonstrate that by allocating costs based on $d_i$ instead of $d_iT_i$, \WPSd overcomes this issue while also achieving a better PoS than the equal-division rule.

Although the upper bound on PoS of \WPSd in Corollary~\ref{cor:WPSd-pos} is derived from a general bound applicable to arbitrary \WPSshort rules, we show in Proposition~\ref{prop:h-unknow-POS>gamma} that, surprisingly, this bound is essentially tight.
Specifically, we prove that when holding costs are private, no \WPSshort rule can achieve a PoS better than $\Omega(\sqrt{\ln \gamma_d})$.
This result implies that \WPSd attains the optimal PoS up to a constant factor when holding costs are private.
% The independence from private information and its near-optimal PoS make \WPSd particularly attractive for real-world applications.

\begin{proposition}
\label{prop:h-unknow-POS>gamma}
When $h_i$ is private information, the PoS of any \WPSshort rule is at least~$\frac{\sqrt{\ln\gamma_d}}{3}$, even if $K_i=0$ and $d_i=2$ for every retailer $i \in N$.
\end{proposition}

\subsection{Proportional Sharing by Estimation of $H_i$}
\label{sec:estimation} 

We have shown that when $H_i$ is known, \WPSh achieves a near-optimal PoS of 1.25, whereas when $H_i$ is unknown, the PoS of any \WPSshort rule is at least $\sqrt{\ln\gamma_d}/3$.
This contrast highlights the critical role of $H_i$ (or equivalently $h_i$, since $d_i$ is public) in maintaining efficiency.

In practice, $H_i$ may not be known precisely but can often be estimated.
To analyze this setting, let $\hat{h}_i$ denote an estimate of retailer $i$'s true holding cost rate~$h_i$.
We quantify estimation accuracy via the maximum multiplicative error
$\epsilon_{\max} \coloneqq \max_{i \in N} \max\{\frac{\hat{h}_i}{h_i},\frac{h_i}{\hat{h}_i}\}$, which captures the worst-case ratio between estimates and true values.
Using weights $w_i=\frac{1}{2}\hat{h}_id_i$, the ratio $\gamma_w=\frac{\max_i H_i/w_i}{\min_i H_i/w_i}=\frac{\max_i h_i/\hat{h}_i}{\min_i h_i/\hat{h}_i}$ directly captures the estimation accuracy.
Crucially, we have $\gamma_w \leq \epsilon_{\max}^2$.
We refer to the resulting \WPSshort rule with weights $w_i=\frac{1}{2}\hat{h}_id_i$ as \WPShat.
By Proposition~\ref{prop:jump-bound} and Theorem~\ref{thm:general-POS}, the PoS of \WPShat can be upper bounded in terms of $\epsilon_{\max}$.

\begin{corollary}
\label{cor:estimation}
For \WPShat, the price of stability is bounded by $\min\big\{4\sqrt{\lceil 2\log \epsilon_{\max} \rceil}+\frac{1}{\sqrt{2}}+1, \  2\sqrt{1+\epsilon_{\max}^2}\big\}$.
% \[\min\big\{4\sqrt{\lceil 2\log \epsilon_{\max} \rceil}+\frac{1}{\sqrt{2}}+1, \  2\sqrt{1+\epsilon_{\max}^2}\big\}.\]
\end{corollary}

This result implies that even when the exact value of $H_i$ is unknown, as long as the estimation error is not big, the efficiency loss caused by incomplete information remains very well controlled.

\section{Conclusion}

We have investigated a broad class of \WPSshort rules in the non-cooperative game-theoretical joint replenishment model, with the goal of improving efficiency of decentralized systems.
After establishing both the existence and efficient compatibility of payoff dominant Nash equilibria for all rules in this class, 
we proposed a suite of theoretically optimal or near-optimal cost allocation rules tailored to various practical and informational constraints.
Specifically, we introduced \WPSh, a simple and practical rule that avoids the drawbacks of free-riders and global information dependence inherent in the theoretically optimal benchmark rule \WPSo, while simultaneously achieving a near-optimal PoS of at most 1.25.
To further accommodate scenarios without access to private information, we proposed two additional practical rules: \WPSd, which relies only on publicly available retailer demand rates, and \WPShat, which effectively incorporates estimated retailer holding cost rates. 
Both rules ensure strong coordination performance under realistic information constraints.
Moreover, we complemented these positive results by providing lower bounds on the best achievable PoS under corresponding informational constraints, thereby providing a complete characterization of efficiency limits for decentralized joint replenishment systems.
A promising direction for future research is to extend our results to joint replenishment models with general setup costs, 
particularly those involving submodular setup costs~\citep{he2012polymatroid,federgruen1992joint,zhang_cost_2009}.

% Bibliography
\bibliography{JRP}

\newpage
% Appendix
\appendix
\section{Missing Proofs in Section \ref{sec:WPS}}

\subsection{Proof of Lemma \ref{lem:no-back-jump}}
\label{app:lem:no-back-jump}

\begin{proof}
According to Lemma \ref{lem:right-jump}, $f_i(2T^1_i;T^1_{-i}) < f_i(T^1_i;T^1_{-i})$ implies that
\begin{align*}
    T^2_i=T^1_i < \sqrt{\frac{1}{2}\frac{K_i+\frac{w_i}{\sum_{j \in N[T^1_i;T^1_{-i}]}w_j}K_0}{H_i}} 
     \leq  \sqrt{\frac{1}{2}\frac{K_i+\frac{w_i}{\sum_{j \in N[T^2_i;T^2_{-i}]}w_j}K_0}{H_i}},
\end{align*}
which implies that $f_i(2T^2_i;T^2_{-i}) < f_i(T^2_i;T^2_{-i})$.
\hfill$\square$ \end{proof}

\subsection{Proof of Lemma \ref{lem:left-jump}}
\label{app:lem:left-jump}

\begin{proof}
The proof is analogous to the proof of Lemma \ref{lem:right-jump}.
When retailer $i$ decreases its replenishment interval from $T_i$ to $T_i/2$, its holding cost is decreased by $H_iT_i/2$ while its setup cost is increased by 
\[
\frac{K_i+\frac{w_i}{\sum_{j \in N[T_i/2;T_{-i}]}w_j}K_0}{T^c_i}.
\]
This change reduces retailer $i$'s total cost if and only if
\[
H_iT_i/2>\frac{K_i+\frac{w_i}{\sum_{j \in N[T_i/2;T_{-i}]}w_j}K_0}{T_i}, 
\]
or
\[
T_i> \sqrt{2\frac{K_i+\frac{w_i}{\sum_{j \in N[T_i/2;T_{-i}]}w_j}K_0}{H_i}}.
\]
Thus statements (2) and (3) are equivalent. 

Since (2) directly implies (1), it remains to show that (1) implies (2).
Suppose that retailer $i$ can reduce its cost by decreasing its replenishment interval from $T_i$ to $2^{-z}T_i$ for some $z \in \mathbb{Z^+}$.
After this change, the holding cost of retailer $i$ is decreased by $(1-2^{-z})H_iT_i$ while the setup cost is increased by at least
\[
\frac{K_i+\frac{w_i}{\sum_{j \in N[2^{-z}T_i;T_{-i}]}w_j}K_0}{2^{1-z}T_i}.
\]
Then
\[
(1-2^{-z})H_iT_i>\frac{K_i+\frac{w_i}{\sum_{j \in N[2^{-z}T_i;T_{-i}]}w_j}K_0}{2^{1-z}T_i}, 
\]
or
\[
T_i>\sqrt{\frac{2^{z-1}}{1-2^{-z}}\frac{K_i+\frac{w_i}{\sum_{j \in N[2^{-z}T_i;T_{-i}]}w_j}K_0}{H_i}} \geq \sqrt{2\frac{K_i+\frac{w_i}{\sum_{j \in N[T_i/2;T_{-i}]}w_j}K_0}{H_i}},
\]
% \begin{align*}
% T_i &> \sqrt{\frac{2^{z-1}}{1-2^{-z}}
%          \frac{K_i+\frac{w_i}{\sum_{j \in N[2^{-z}T_i;\;T_{-i}]}w_j}\,K_0}{H_i}}\\
%     &\ge \sqrt{2\,\frac{K_i+\frac{w_i}{\sum_{j \in N[T_i/2;\;T_{-i}]}w_j}\,K_0}{H_i}}.
% \end{align*}
which means that retailer $i$ can also reduce its cost by halving its replenishment interval.
Thus, (1) implies (2).
\hfill$\square$ \end{proof}

\subsection{Proof of Lemma \ref{lem:no-left-jump}}
\label{app:lem:no-left-jump}

\begin{proof}
Suppose, for contradiction,  that retailer $i$ can reduce its cost by decreasing its replenishment interval from $T^c_i$ to $2^{-z}T^c_i$ for some $z \in \mathbb{Z^+}$.
If $i \in U$, then $T^c_i  \leq \sqrt{2s}$, and $N[2^{-z}T^c_i;T_{-i}]=\{i\}$.
According to Lemma \ref{lem:left-jump}, we have
\[
\sqrt{2s}\geq T^c_i >\sqrt{2\frac{K_i+K_0}{H_i}},
\]
which contradicts with the definition of $s$ in (\ref{eq:s-def}).
 % $s= \min_{\emptyset \subset S \subseteq N}\frac{K_0+K(S)}{H(S)}$.
If $i \in V$, then $T^c_i  \leq \sqrt{2\frac{K_i}{H_i}}$.
According to Lemma \ref{lem:left-jump}, we have
\[
\sqrt{2\frac{K_i}{H_i}}\geq T^c_i >\sqrt{2\frac{K_i+\frac{w_i}{\sum_{j \in N[2^{-z}T^c_i;T_{-i}]}w_j}K_0}{H_i}},
\]
which is a contradiction since $w_i \geq 0$.
Thus we get a contradiction for any retailer $i \in N$. 
\hfill$\square$ \end{proof}

\subsection{Proof of Lemma \ref{lem:alg1-NE}}
\label{app:lem:alg1-NE}

\begin{proof}
Since no retailer halves its replenishment interval (Lemma~\ref{lem:no-left-jump-ALG}) and retailer $i$'s strategy set is 
\[
\Gamma_i=\{T_i \mid \sqrt{K_i / (2H_i)} \le T_i \le \sqrt{2(K_0+K_i) / H_i} \text{ such that } T_i=2^{z_i}B, z_i \in \mathbb{Z}\},
\]
in Algorithm~1 retailer $i$ can make at most $\log_2 \sqrt{4\frac{K_0+K_i}{K_i}}$ number of changes.
Overall, there are at most 
% $n\log_2 \sqrt{4\left(1+\frac{K_0}{\min K_i}\right)}$ 
$O\left(n\log_2\left(1+\frac{K_0}{\min_i K_i}\right)\right)$
changes.
In each round checking every $i \in N$, at least one retailer doubles its replenishment interval.
So we get the running time $O\left(n^2\log_2\left(1+\frac{K_0}{\min_i K_i}\right)\right)$.
Since no retailer wants to double or halve its replenishment interval at the final state, according to Lemma~\ref{lem:right-jump} and Lemma~\ref{lem:left-jump}, the final state is a Nash equilibrium. 
\hfill$\square$ \end{proof}

\subsection{Proof of Proposition \ref{prop:POA-bound-upper}}
\label{app:prop:POA-bound-upper}

\begin{proof}
% By property P2, 
The share of the major setup cost for any retailer $i$ is at most $\frac{K_0}{T_i}$ under any joint replenishment policy $T$, i.e., 
\[
x_i(T_i;T_{-i}) \leq \frac{K_0}{T_i}, \forall i \in N, \forall T \in \Gamma.
\]
Let $T^*$ be an arbitrary Nash equilibrium. We have
\begin{align*}
C(T^*) &=\sum_{i=1}^n f_i(T_i^*;T_{-i}^*) \leq \sum_{i=1}^n \min_{T_i \in \Gamma_i} f_i(T_i;T_{-i}^*) \\
&= \sum_{i=1}^n \min_{T_i \in \Gamma_i} \left(H_iT_i+\frac{K_i}{T_i}+x_i(T_i^*;T_{-i}^*)\right) \\
&\leq \sum_{i=1}^n \min_{T_i \in \Gamma_i} \left(H_iT_i+\frac{K_i+K_0}{T_i}\right). 
\end{align*}
On the other hand, for the social optimum we have
\begin{align*}
C(T^c) &=\frac{K_0}{T_{\min}^c}+\sum_{i=1}^n \left(H_iT_i^c+\frac{K_i}{T_i^c}\right) \\
&\geq \sum_{i=1}^n \left(H_iT_i^c+\frac{K_i+K_0/n}{T_i^c}\right) \\
&\geq \sum_{i=1}^n \min_{T_i \in \Gamma_i} \left(H_iT_i+\frac{K_i+K_0/n}{T_i}\right). 
% &\geq \sum_{i=1}^n \frac{3}{2\sqrt{2}} \cdot 2 \cdot \sqrt{H_i(K_i+K_0/n)} 
\end{align*}
So the PoA is at most 
\begin{align*}
    \frac{C(T^*)}{C(T^c)} &\leq \frac{\sum_{i=1}^n \min_{T_i \in \Gamma_i} \left(H_iT_i+\frac{K_i+K_0}{T_i}\right)}{\sum_{i=1}^n \min_{T_i \in \Gamma_i} \left(H_iT_i+\frac{K_i+K_0/n}{T_i}\right)}  
    \leq \frac{3}{2\sqrt{2}}\sqrt{n},
\end{align*}
where the last inequality follows from the fact that $2\sqrt{ab} \leq \min_{y\in\{2^z \mid z \in \mathbb{Z}\}} ay+\frac{b}{y} \leq \frac{3}{\sqrt{2}}\sqrt{ab}$ for any $a,b \geq 0$.
\hfill$\square$ \end{proof}

\subsection{Proof of Proposition \ref{prop:POA-bound}}
\label{app:prop:POA-bound}

\begin{proof}
Let $K_0=1$. Let $K_i=0$, $h_i=1$, $d_i=2$, and $H_i=\frac{1}{2}h_id_i=1$ for each retailer $i \in N$. 
For the optimal centralized policy $T^c$ we have $s=\frac{1}{n}$ and $C(T^c) \leq (\sqrt{2}+\frac{1}{\sqrt{2}})\sqrt{n}=\frac{3}{\sqrt{2}}\sqrt{n}$.
On the other hand, for any rule satisfying symmetry, there is a Nash equilibrium $\tilde{T}$ with $\tilde{T}_i=1$.
Indeed, the cost of every retailer under $\tilde{T}$ is $H_i\tilde{T}_i+\frac{1}{n}\frac{K_0}{\tilde{T}_i}=1+\frac{1}{n}$ by symmetry.
If a retailer doubles its replenishment interval, its cost will be at least $2H_i\tilde{T}_i=2>1+\frac{1}{n}$.
If a retailer halves its replenishment interval, its cost will be at least $\frac{1}{2}H_i\tilde{T}_i+\frac{1}{2}\frac{K_0}{\tilde{T}_i/2}=1.5 \geq 1+\frac{1}{n}$.
 % by property P2.
Since $C(\tilde{T})=n+1$, the PoA is at least $\frac{n+1}{\frac{3}{\sqrt{2}}\sqrt{n}} \geq \frac{\sqrt{2}}{3}\sqrt{n}$.
\hfill$\square$ \end{proof}

\newpage
\section{Missing Proofs in Section \ref{sec:efficient-rules}}

\subsection{Proof of Proposition \ref{prop:k-know-POS=1}}
\label{app:prop:k-know-POS=1}

\begin{proof}
% Given an instance $(K_0,K_i,H_i)$, we compute the partition $N=U \cup V$ and set $w_i=0$ for each $i \in V$.
% For retailers in $U$, we set $w_i$ such that under the optimal centralized policy $T^c$ the cost of each retailer $i$ satisfies that
% \[
% f_i(T^c_i,T^c_{-i})=H_iT^c_i+\frac{K_i}{T^c_i}+x_i(T^c_i,T^c_{-i})=\frac{H_i}{\sum_{j \in U}H_j}(\sum_{j \in U}H_jT^c_{\min}+\frac{K_0+\sum_{j \in U}K_j}{T^c_{\min}}),
% \]
% or 
% \begin{align}
% \label{eq:rule-H+K}
% \frac{K_i}{T^c_{\min}}+x_i(T^c_i,T^c_{-i})=\frac{H_i}{\sum_{j \in U}H_j}\frac{K_0+\sum_{j \in U}K_j}{T^c_{\min}}.
% \end{align}
% Notice that 
% \[
% x_i(T^c_i,T^c_{-i})=\frac{1}{\sum_{j \in U}H_jT^c_{\min}}((K_0+\sum_{j \in U}K_j)H_i-\sum_{j \in U}H_jK_i)) \geq 0.
% \]
% Thus, we need to set
% \[
% w_i=\frac{x_i(T^c_i,T^c_{-i})}{\frac{K_0}{T^c_{\min}}}=
% \frac{1}{\sum_{j \in U}H_jK_0}((K_0+\sum_{j \in U}K_j)H_i-\sum_{j \in U}H_jK_i))
% =\frac{H_is^2-K_i}{K_0}.
% \]
We show that the optimal centralized policy is a Nash equilibrium under the defined rule.
Given an instance profile $(K_0,N,\{K_i\}_{i \in N},\{H_i\}_{i \in N})$, we compute the partition $N=U \cup V$ according to (\ref{eq:U-V}).
% and set~$w^*_i$ as in (\ref{optimal-w}) for each $i \in N$.
Due to the selected value for $w^*_i$, for any retailers $i \in U$, we have
\begin{align}
\frac{K_i}{T^c_i}+x_i(T^c_i;T^c_{-i})
&=\frac{K_i}{T^c_{\min}}+\frac{w^*_i}{\sum_{j \in U}w^*_j}\frac{K_0}{T^c_{\min}} \nonumber\\
&=\frac{K_i}{T^c_{\min}}+\frac{H_is-K_i}{K_0}\frac{K_0}{T^c_{\min}} \nonumber\\
&=\frac{H_is}{T^c_{\min}} \nonumber\\
&=\frac{H_i}{\sum_{j \in U}H_j}\frac{K_0+\sum_{j \in U}K_j}{T^c_{\min}}.\label{eq:rule-H+K-setup}
\end{align}
Notice that $\frac{K_0+\sum_{j \in U}K_j}{T^c_{\min}}$ is the total setup cost for all retailers in $U$, consisting of the major setup cost and the minor setup costs of retailers in $U$.
The above equation shows that, under \WPSo, retailers in $U$ share their total setup cost in proportion to their $H_i$.
Since $T_i^c=T_{\min}^c$ for each $i \in U$, it follows that
\begin{align}\label{eq:rule-H+K-holding}
f_i(T^c_i,T^c_{-i})
& =H_iT^c_{\min}+\frac{K_i}{T^c_{\min}}+x_i(T^c_i;T^c_{-i}) \nonumber\\
& \overset{(\ref{eq:rule-H+K-setup})}{=}\frac{H_i}{\sum_{j \in U}H_j}\left(\sum_{j \in U}H_jT^c_{\min}+\frac{K_0+\sum_{j \in U}K_j}{T^c_{\min}}\right).
\end{align}
That is, retailers in $U$ share their total cost in proportion to their $H_i$.
In addition, for each retailer $i \in U$, the ratio between its holding cost and its setup cost (including its minor setup cost and its share of the major setup cost) is the same as the ratio for the whole set $U$.
Since $T^c_{\min}$ minimizes the total cost for $U$, we can show that no retailer from $U$ can reduce its cost by doubling its replenishment interval.
Indeed, if retailer $i \in U$ doubles its replenishment interval, then its share for the major setup cost will be
\begin{align*}
x_i(2T^c_i,T^c_{-i}) 
&=\frac{w^*_i}{\sum_{j \in N[2T_i^c;T_{-i}^c]}w^*_j}\frac{K_0}{2T^c_{\min}} \\
&=\frac{w^*_i}{\sum_{j \in U}w^*_j}\frac{K_0}{2T^c_{\min}} \\
&=\frac{x_i(T^c_i,T^c_{-i})}{2},
\end{align*}
where the second equation follows by that $w^*_j=0$ for any $j \in N[2T_i^c;T_{-i}^c] \setminus U \subseteq V$.
Then its overall cost will be
\begin{align*}
f_i(2T^c_i,T^c_{-i})
&=2H_iT^c_i+\frac{K_i}{2T^c_i}+x_i(2T^c_i,T^c_{-i}) \\
&=2H_iT^c_i+\frac{K_i}{2T^c_i}+\frac{x_i(T^c_i,T^c_{-i})}{2} \\
&=
\frac{H_i}{\sum_{j \in U}H_j}\left(2\sum_{j \in U}H_jT^c_{\min}+\frac{K_0+\sum_{j \in U}K_j}{2T^c_{\min}}\right) \\
&\geq \frac{H_i}{\sum_{j \in U}H_j}C(T^c) \overset{(\ref{eq:rule-H+K-holding})}{=}f_i(T^c_i,T^c_{-i}),
\end{align*}
where the last inequality follows from that $T^c$ is an optimal centralized policy.
So no retailer from $U$ can benefit by doubling its replenishment interval.
According to Lemma \ref{lem:no-left-jump}, no retailer from $U$ can reduce its cost by decreasing its replenishment interval.
For each retailer in $V$, since its total cost already matches its optimal cost without sharing $K_0$, so it has no incentive to change.
Therefore, $T^c$ is a Nash equilibrium and hence the PoS is 1.
\hfill$\square$ \end{proof}

\subsection{Proof of Lemma \ref{lem:no-double}}
\label{app:lem:no-double}

\begin{proof}
Denote $P= \{j \in N \mid j \leq i \text{ and } T^c_j=T^c_i\}$.
We show that $T^h_i=T^c_i$ implies $T^h_{i-1}=T^c_{i-1}$.
Then applying the same argument for each pair of consecutive retailers from $P$ we will get $T^h_j=T^c_j$ for any $j \in P$.
Since $T^h_i=T^c_i$, we know that retailer $i$ does not double its replenishment interval under policy $T^{i+1}$.
According to Lemma \ref{lem:right-jump}, we have
\[
T^c_i=T^{i+1}_i \geq \sqrt{\frac{1}{2}\left(\frac{K_i}{H_i}+\frac{K_0}{\sum_{j \in N[T^{i+1}_i;T^{i+1}_{-i}]}H_j}\right)}.
\]
Since $T^i_{i-1}=T^i_i$ (by the definition of $P$ and $i-1 \in P$) and $T^{i+1}=T^i$ (since $i$ does not double its replenishment interval), we have $N[T^{i+1}_i;T^{i+1}_{-i}] = N[T^{i}_i;T^{i}_{-i}] = N[T^{i}_{i-1};T^{i}_{-(i-1)}]$.
Then
\begin{align*}
    T^c_{i-1}=T^c_i &\geq \sqrt{\frac{1}{2}\left(\frac{K_i}{H_i}+\frac{K_0}{\sum_{j \in N[T^{i+1}_i;T^{i+1}_{-i}]}H_j}\right)}\\
    & \geq \sqrt{\frac{1}{2}\left(\frac{K_{i-1}}{H_{i-1}}+\frac{K_0}{\sum_{j \in N[T^{i}_{i-1};T^{i}_{-(i-1)}]}H_j}\right)}.
\end{align*}
By Lemma \ref{lem:right-jump}, retailer $i-1$ cannot reduce its cost by doubling its replenishment interval under $T^i$, so $T^h_{i-1}=T^c_{i-1}$.
\hfill$\square$ \end{proof}

\subsection{Proof of Lemma \ref{lem:same-N}}
\label{app:lem:same-N}

\begin{proof}
We partition the retailer set $N$ into four subsets:
$S^1=\{j \in N \mid j \leq i \text{ and } T^i_j<T^i_i\}$, $S^2=\{j \in N \mid j \leq i \text{ and } T^i_j=T^i_i\}$, $S^3=\{j \in N \mid j > i \text{ and } T^i_j=T^i_i\}$, and $S^4=\{j \in N \mid j > i \text{ and } T^i_j> T^i_i\}$.
By definition, $N[T^i_i;T^i_{-i}]=S^1 \cup S^2 \cup S^3$.
Since the replenishment intervals of retailers from $S^3 \cup S^4$ do not change from $T^i$ to $T^h$, we have $S^3 \subseteq N[T^h_i;T^h_{-i}]$ and $S^4 \cap N[T^h_i;T^h_{-i}]=\emptyset$.
Then it suffices to show $S^1 \cup S^2 \subseteq N[T^h_i;T^h_{-i}]$, or equivalently, $T^h_j \leq T^h_i$ for any $j \in S^1 \cup S^2$. 
% Denote $P= \{j \in N \mid j \leq i \text{ and } T^c_j=T^c_i\}\subseteq S^1$.
Since $T^h_i=T^c_i$, we have $T^h_i=T^c_i=T^i_i$.
According to Lemma \ref{lem:no-double} and $T^h_i=T^c_i$, we have $T^h_j=T^i_j=T^i_i=T^h_i$ for any $j \in S^2$.
For any $j \in S^1$, from $T^i_j<T^i_i$ we have $T^i_j \leq \frac{1}{2}T^i_i= \frac{1}{2}T^h_i$.
Since every retailer doubles its replenishment interval at most once during Algorithm~2, we have $T^h_j \leq 2T^i_j \leq T^h_i$ for any $j \in S^1$.
\hfill$\square$ \end{proof}

\subsection{Proof of Theorem \ref{thm:NE-alg}}
\label{app:thm:NE-alg}

\begin{proof}
The running time of Algorithm~2 is clear.
We first show that $T^h$ is a Nash equilibrium.
According to Lemma~\ref{lem:right-jump} and Lemma~\ref{lem:left-jump}, it suffices to show that no retailer has an incentive to double or halve its replenishment interval under $T^h$.

For any retailer $i \in N$ with $T^h_i=T^c_i$, according to Lemma~\ref{lem:no-left-jump}, retailer cannot reduce its cost by decreasing its replenishment interval under $T^h$.
For the other direction, it follows from $T^h_i=T^c_i$ that retailer $i$ cannot reduce its cost by doubling its replenishment interval under $T^{i+1}=T^i$.
Since $N[T^i_i;T^i_{-i}]=N[T^h_i;T^h_{-i}]$ (Lemma~\ref{lem:same-N}), according to Lemma~\ref{lem:no-back-jump}, retailer $i$ cannot reduce its cost by doubling its replenishment interval under $T^h$.
% Then by Lemma \ref{lem:right-jump}, it cannot reduce its cost by increasing its replenishment interval under $T^h$.

For any retailer $i \in N$ with $T^h_i=2T^c_i$, since retailer $i$ doubles its replenishment interval under $T^{i+1}$, we have 
$
f_i(2T^{i+1}_i;T^{i+1}_{-i}) < f_i(T^{i+1}_i;T^{i+1}_{-i})
$.
Since $T^h \geq T^{i+1}$, we have $N[T^{i+1}_i;T^h_{-i}] \subseteq N[T^{i+1}_i;T^{i+1}_{-i}]$.
According to Lemma \ref{lem:no-back-jump}, we get
$
f_i(2T^{h}_i;T^{h}_{-i}) < f_i(T^{h}_i;T^{h}_{-i})
$,
which means retailer $i$ cannot reduce its cost by halving its replenishment interval under $T^h$.
% Together with Lemma \ref{lem:no-left-jump} we have that retailer $i$ cannot reduce its cost by decreasing its replenishment interval under $T^h$.
For the other direction, for any retailer $i \in N$ with $T^h_i=2T^c_i$, we have $U \subseteq N[T^h_i;T^h_{-i}]$.
If $i \in U$, then
\begin{align*}
    T^h_i =2T^c_i=2T^c_{\min}  \geq \sqrt{2\frac{K_0+\sum_{j \in U}K_j}{\sum_{j \in U}H_j}} 
    \geq \sqrt{\frac{K_i}{H_i}+\frac{K_0}{\sum_{j \in U}H_j}}
    \geq \sqrt{\frac{K_i}{H_i}+\frac{K_0}{\sum_{j \in N[T^h_i;T^h_{-i}]}H_j}}.
\end{align*}
Similarly, if $i \in V$, then
\begin{align*}
    T^h_i=2T^c_i\geq \sqrt{2\frac{K_i}{H_i}}  \geq \sqrt{\frac{K_i}{H_i}+\frac{K_0+\sum_{j \in U}K_j}{\sum_{j \in U}H_j}}  \geq \sqrt{\frac{K_i}{H_i}+\frac{K_0}{\sum_{j \in N[T^h_i;T^h_{-i}]}H_j}}.
\end{align*}
Then according to Lemma \ref{lem:right-jump}  and $w_j=H_j$ for every $j \in N$, retailer $i$ cannot reduce its cost by doubling its replenishment interval under $T^h$.

To show that $T^h$ is payoff dominant, notice that the proof of Theorem \ref{thm:general-NE-exist} does not depend on the update order of retailers in Algorithm~1.
Furthermore, no retailer ever halves its replenishment interval during Algorithm~1.
Therefore, Algorithm~2 can be seen as a special version of Algorithm~1 with a special update order.
Then, according to Theorem \ref{thm:general-NE-exist}, $T^h$ found by Algorithm~2 is a payoff dominant Nash equilibrium.
\hfill$\square$ \end{proof}

\subsection{Proof of Lemma \ref{lem:POT-left}}
\label{app:lem:POT-left}

\begin{proof}
We show the claim for $i \in U$. The other case for $i \in V$ is analogous.
According to Lemma~\ref{lem:no-double}, if there exists some retailer $i \in U$ such that $T^h_i=2T^c_i$, then it must holds that $T^h_{i^*}=2T^c_{i^*}$, where $i^*$ is the first retailer from $U$ considered by Algorithm~2.
Thus it suffices to prove the claim for retailer $i^*$.
Since $T^c_{i^*}=T^c_{\min} \overset{POT}{=} \sqrt{s}$, we have $\alpha \in [\frac{1}{\sqrt{2}},\sqrt{2}]$.
Suppose that $\alpha \in [1,\sqrt{2}]$, then 
\begin{align*}
    T^c_{i^*} \geq \sqrt{s}&=\sqrt{\frac{K_0+\sum_{i \in U}K_j}{\sum_{i \in U}H_j}}
    \geq \sqrt{\frac{1}{2}\left(\frac{K_{i^*}}{H_{i^*}}+\frac{K_0}{\sum_{i \in U}H_j}\right)}.
\end{align*}
On the other hand, since $T^h_{i^*}=2T^c_{i^*}$, retailer $i^*$ doubles its replenishment interval during Algorithm~2.
By the choice of $i^*$, the set of retailers with replenishment interval at most $T^c_{i^*}$ just before the doubling is $U$. 
Then according to Lemma \ref{lem:right-jump}, we have
\[
T^c_{i^*}<
% \sqrt{\frac{1}{2}(\frac{K_{i^*}}{H_{i^*}}+\frac{1}{\sum_{j \in N[T^{i^*+1}_{i^*};T^{i^*+1}_{-i^*}]}H_j}K_0)}=
\sqrt{\frac{1}{2}\left(\frac{K_{i^*}}{H_{i^*}}+\frac{K_0}{\sum_{i \in U}H_j}\right)},
\]
which is a contradiction.
\hfill$\square$ \end{proof}

\subsection{Proof of Theorem \ref{thm:POS}}
\label{app:thm:POS}

\begin{proof}
We divide the total cost under $T^h$ into two parts
\begin{align*}
    C(T^h)= \left(\frac{K_0}{T^h_{\min}}+\sum_{j \in U}\left(H_jT^h_j+\frac{K_j}{T^h_j}\right)\right)+  \sum_{j \in V}\left(H_jT^h_j+\frac{K_j}{T^h_j}\right).
\end{align*}
We show that each component of $C(T^h)$ can be upper bounded by 1.25 times the corresponding component under the optimal centralized policy $T^c$.

We first consider the second part of set $V$.
For any retailer $j\in V$, let $T^c_j=\alpha_j \sqrt{\frac{K_j}{H_j}}$.
If $T^h_j=T^c_j$, then $H_jT^h_j+\frac{K_j}{T^h_j}$ remains the same as under $T^c$.
Thus it suffices to consider the case with $T^h_j=2T^c_j$, where $\alpha_j \in [\frac{1}{\sqrt{2}},1)$ according to Lemma \ref{lem:POT-left}.
Notice that $H_j\sqrt{\frac{K_j}{H_j}}=\frac{K_j}{\sqrt{\frac{K_j}{H_j}}}$.
% \[
% H_j\sqrt{\frac{K_j}{H_j}}=\frac{K_j}{\sqrt{\frac{K_j}{H_j}}}.
% \]
Then
\begin{align}\label{eq:wpsh-pos-relax-4}
    \frac{H_jT^h_j+\frac{K_j}{T^h_j}}{H_jT^c_j+\frac{K_j}{T^c_j}}
&=\frac{2H_jT^c_j+\frac{K_j}{2T^c_j}}{H_jT^c_j+\frac{K_j}{T^c_j}}
=\frac{2\alpha_j H_j\sqrt{\frac{K_j}{H_j}}+\frac{K_j}{2\alpha_j \sqrt{\frac{K_j}{H_j}}}}{\alpha_j H_j\sqrt{\frac{K_j}{H_j}}+\frac{K_j}{\alpha_j \sqrt{\frac{K_j}{H_j}}}} \nonumber \\
&=\frac{2\alpha_j+\frac{1}{2\alpha_j}}{\alpha_j+\frac{1}{\alpha_j}}
=2\frac{\alpha_j^2+\frac{1}{4}}{\alpha_j^2+1}
\leq \frac{5}{4},
\end{align}
where the last inequality follows from $\alpha_j<1$.

Next we analyze the first part $\frac{K_0}{T^h_{\min}}+\sum_{j \in U}\left(H_jT^h_j+\frac{K_j}{T^h_j}\right)$.
Let $T^c_{\min}=\alpha \sqrt{s}$, where $s=\frac{K_0+\sum_{i \in U}K_j}{\sum_{i \in U}H_j}$.
According to Lemma \ref{lem:POT-left}, if $\alpha \geq 1$, then $T^h_i=T^c_i$ for every $i \in U$ and hence the first part remains the same as under $T^c$.
Thus it suffices to consider the case with $\alpha \in [\frac{1}{\sqrt{2}},1)$.
Notice that 
\begin{align}\label{eq:wpsh-pos-equal}
\sum_{j \in U}H_j\sqrt{s}=\frac{K_0+\sum_{j \in U}K_j}{\sqrt{s}}.
\end{align}
Let $U_0\subseteq U$ be the set of retailers with $T^h_i=T^c_{\min}$.
If $U_0=\emptyset$, then each retailer $i \in U$ satisfies $T^h_i=2T^c_i$, and the analysis will be similar as before for retailers in $V$.
When $U_0 \neq \emptyset$, we require an additional step to handle the retailers in $U_0$, as detailed below.
For any retailer $i \in U_0$, if it doubles its replenishment interval, its holding cost will be increased by $H_iT^c_{\min}$ while its setup cost will be reduced by $\frac{H_i}{\sum_{j \in U_0}H_j}\frac{K_0}{2T^c_{\min}}+\frac{K_i}{2T^c_{\min}}$.
Since $i \in U_0$, we have $\frac{H_i}{\sum_{j \in U_0}H_j}\frac{K_0}{2T^c_{\min}}+\frac{K_i}{2T^c_{\min}}<H_iT^c_{\min}$.
% \[
% \frac{1}{|U_0|}\frac{K_0}{2T^c_{\min}}+\frac{K_i}{2T^c_{\min}}<H_iT^c_{\min}.
% \]
Summing them up for all agent in $U_0$ we get
\begin{align}\label{eq:wpsh-pos-relax-1}
\frac{K_0}{2T^c_{\min}}+\sum_{i \in U_0}\frac{K_i}{2T^c_{\min}}<\sum_{i \in U_0}H_iT^c_{\min}.
\end{align}
Then we can upper bound the first part as follows:
\begin{align}\label{eq:wpsh-pos-relax-2}
\frac{K_0}{T^h_{\min}}+\sum_{j \in U}\left(H_jT^h_j+\frac{K_j}{T^h_j}\right)
& = \frac{K_0}{2T^c_{\min}} +\sum_{j \in U_0}\left(H_jT^c_{\min}+\frac{K_j}{T^c_{\min}}\right)
+ \frac{K_0}{2T^c_{\min}} +\sum_{j \in U \setminus U_0}\left(2H_jT^c_{\min}+\frac{K_j}{2T^c_{\min}}\right) \nonumber \\ 
& < \sum_{j \in U_0}\left(2H_jT^c_{\min}+\frac{K_j}{2T^c_{\min}}\right) + \frac{K_0}{2T^c_{\min}} +\sum_{j \in U \setminus U_0}\left(2H_jT^c_{\min}+\frac{K_j}{2T^c_{\min}}\right) \nonumber \\
& = \sum_{j \in U}2H_jT^c_{\min}+\frac{K_0+\sum_{j \in U}K_j}{2T^c_{\min}}, 
\end{align}
% \begin{align}\label{eq:wpsh-pos-relax-2}
% & \frac{K_0}{T^h_{\min}}
%  \;+\; \sum_{j \in U} \biggl(H_j T^h_j + \frac{K_j}{T^h_j}\biggr) \nonumber \\
% =\;& \frac{K_0}{2T^c_{\min}}
%   \;+\; \sum_{j \in U_0} \biggl(H_j T^c_{\min} + \frac{K_j}{T^c_{\min}}\biggr) \nonumber  \\
%   & \;+\; \frac{K_0}{2T^c_{\min}} 
%  +\; \sum_{j \in U \setminus U_0} \biggl(2H_j T^c_{\min} 
%   + \frac{K_j}{2T^c_{\min}}\biggr) \nonumber \\
% <\;& \sum_{j \in U_0} \biggl(2H_j T^c_{\min}  
%   + \frac{K_j}{2T^c_{\min}}\biggr) \nonumber \\
%   & \;+\; \frac{K_0}{2T^c_{\min}}  
% +\; \sum_{j \in U \setminus U_0} \biggl(2H_j T^c_{\min}
%   + \frac{K_j}{2T^c_{\min}}\biggr) \nonumber \\
% =\;& \sum_{j \in U} 2H_j T^c_{\min}
%   \;+\; \frac{K_0 + \sum_{j \in U} K_j}{2T^c_{\min}}.
% \end{align}
% \begin{align}\label{eq:wpsh-pos-relax-2}
% \frac{K_0}{T^h_{\min}}
% & \;+\; \sum_{j \in U} \biggl(H_j T^h_j + \frac{K_j}{T^h_j}\biggr) \nonumber \\
% &=\; \frac{K_0}{2T^c_{\min}}
%   \;+\; \sum_{j \in U_0} \biggl(H_j T^c_{\min} + \frac{K_j}{T^c_{\min}}\biggr) \nonumber \\
%   & \quad \;+\; \frac{K_0}{2T^c_{\min}} 
% +\; \sum_{j \in U \setminus U_0} \biggl(2H_j T^c_{\min}
%   + \frac{K_j}{2T^c_{\min}}\biggr) \nonumber \\
% &<\; \sum_{j \in U_0} \biggl(2H_j T^c_{\min}
%   + \frac{K_j}{2T^c_{\min}}\biggr) \nonumber \\
%  & \quad \;+\; \frac{K_0}{2T^c_{\min}} 
%  +\; \sum_{j \in U \setminus U_0} \biggl(2H_j T^c_{\min}
%   + \frac{K_j}{2T^c_{\min}}\biggr) \nonumber \\
% &=\; \sum_{j \in U} 2H_j T^c_{\min}
%   \;+\; \frac{K_0 + \sum_{j \in U} K_j}{2T^c_{\min}}.
% \end{align}
where the inequality follows from (\ref{eq:wpsh-pos-relax-1}).
Now we have that
\begin{align}\label{eq:wpsh-pos-relax-3}
\frac{\frac{K_0}{T^h_{\min}}+\sum_{j \in U}\left(H_jT^h_j+\frac{K_j}{T^h_j}\right)}{\frac{K_0}{T^c_{\min}}+\sum_{j \in U}\left(H_jT^c_j+\frac{K_j}{T^c_j}\right)}
& \overset{(\ref{eq:wpsh-pos-relax-2})}{<}\frac{\sum_{j \in U}2H_jT^c_{\min}+\frac{K_0+\sum_{j \in U}K_j}{2T^c_{\min}}}{\sum_{j \in U}H_jT^c_{\min}+\frac{K_0+\sum_{j \in U}K_j}{T^c_{\min}}} \nonumber \\ 
&=\frac{\sum_{j \in U}2\alpha H_j\sqrt{s}+\frac{K_0+\sum_{j \in U}K_j}{2\alpha \sqrt{s}}}{\sum_{j \in U}\alpha H_j\sqrt{s}+\frac{K_0+\sum_{j \in U}K_j}{\alpha \sqrt{s}}} \nonumber \\
&\overset{(\ref{eq:wpsh-pos-equal})}{=}\frac{2\alpha+\frac{1}{2\alpha}}{\alpha+\frac{1}{\alpha}}
=2\frac{\alpha^2+\frac{1}{4}}{\alpha^2+1}
\leq \frac{5}{4},
\end{align}
% \begin{align}\label{eq:wpsh-pos-relax-3}
% &\frac{\frac{K_0}{T^h_{\min}}
%   + \sum_{j \in U}\!\Bigl(H_j T^h_j 
%   + \frac{K_j}{T^h_j}\Bigr)}%
%   {\frac{K_0}{T^c_{\min}}
%   + \sum_{j \in U}\!\Bigl(H_j T^c_j 
%   + \frac{K_j}{T^c_j}\Bigr)}
% \nonumber \\[6pt]
% &\quad \overset{(\ref{eq:wpsh-pos-relax-2})}{<}\;
%   \frac{\sum_{j \in U} 2H_j T^c_{\min}
%   + \frac{K_0 + \sum_{j \in U} K_j}{2T^c_{\min}}}%
%   {\sum_{j \in U} H_j T^c_{\min}
%   + \frac{K_0 + \sum_{j \in U} K_j}{T^c_{\min}}}
% \nonumber \\[6pt]
% &\quad =\;
%   \frac{\sum_{j \in U} 2\alpha H_j \sqrt{s}
%   + \frac{K_0 + \sum_{j \in U} K_j}{2\alpha \sqrt{s}}}%
%   {\sum_{j \in U} \alpha H_j \sqrt{s}
%   + \frac{K_0 + \sum_{j \in U} K_j}{\alpha \sqrt{s}}}
% \nonumber \\[6pt]
% &\quad \overset{(\ref{eq:wpsh-pos-equal})}{=}\;
%   \frac{2\alpha + \frac{1}{2\alpha}}%
%       {\alpha + \frac{1}{\alpha}}
%       \;\;=\;\; 2\frac{\alpha^2+\frac{1}{4}}{\alpha^2+1}
%   \;\;\le\;\;\frac{5}{4}.
% \end{align}
where the last inequality follows from $\alpha<1$.
% where the last equality follows from $\sum_{j \in U}H_j\sqrt{s}=\frac{K_0+\sum_{j \in U}K_j}{\sqrt{s}}$ and the last inequality from $\alpha<1$.

% then
% \[
% \frac{\frac{K_0}{T^h_{\min}}+\sum_{j \in N}(H_jT^h_j+\frac{K_j}{T^h_j})}{\frac{K_0}{T^c_{\min}}+\sum_{j \in N}(H_jT^c_j+\frac{K_j}{T^c_j})}
% % \frac{\sum_{j \in U}(H_jT^h_j+\frac{K_j}{T^h_j})+\frac{K_0}{T^c_{\min}}}{\sum_{j \in U}H_jT^c_{\min}+\frac{K_0+\sum_{j \in U}K_j}{T^c_{\min}}}
% \leq \frac{\sum_{j \in U}2H_jT^c_{\min}+\frac{K_0+\sum_{j \in U}K_j}{T^c_{\min}}}{\sum_{j \in U}H_jT^c_{\min}+\frac{K_0+\sum_{j \in U}K_j}{T^c_{\min}}}
% =\frac{\sum_{j \in U}2\alpha H_j\sqrt{s}+\frac{K_0+\sum_{j \in U}K_j}{\alpha \sqrt{s}}}{\sum_{j \in U}\alpha H_j\sqrt{s}+\frac{K_0+\sum_{j \in U}K_j}{\alpha \sqrt{s}}}
% =\frac{2\alpha+\frac{1}{\alpha}}{\alpha+\frac{1}{\alpha}}
% \leq \frac{3}{2}.
% \]

Then the ratio between the total cost under policy $T^h$ and that under $T^c$ is bounded from above by 
\begin{align*}
 \frac{C(T^h)}{C(T^c)}
&=\frac{\frac{K_0}{T^h_{\min}}+\sum_{j \in N}\left(H_jT^h_j+\frac{K_j}{T^h_j}\right)}{\frac{K_0}{T^c_{\min}}+\sum_{j \in N}\left(H_jT^c_j+\frac{K_j}{T^c_j}\right)} \\
&=\frac{\frac{K_0}{T^h_{\min}}+\sum_{j \in U}\left(H_jT^h_j+\frac{K_j}{T^h_j}\right)+\sum_{j \in V}\left(H_jT^h_j+\frac{K_j}{T^h_j}\right)}{\frac{K_0}{T^c_{\min}}+\sum_{j \in U}\left(H_jT^c_j+\frac{K_j}{T^c_j}\right)+\sum_{j \in V}\left(H_jT^c_j+\frac{K_j}{T^c_j}\right)} \\
% &\leq \max\left\{\frac{\frac{K_0}{T^h_{\min}}+\sum_{j \in U}\left(H_jT^h_j+\frac{K_j}{T^h_j}\right)}{\frac{K_0}{T^c_{\min}}+\sum_{j \in U}\left(H_jT^c_j+\frac{K_j}{T^c_j}\right)},\\
% & \quad \frac{\sum_{j \in V}\left(H_jT^h_j+\frac{K_j}{T^h_j}\right)}{\sum_{j \in V}\left(H_jT^c_j+\frac{K_j}{T^c_j}\right)}\right\} \\
&\leq\frac{5}{4},
\end{align*}
where the final inequality is obtained by applying inequalities (\ref{eq:wpsh-pos-relax-3}) and (\ref{eq:wpsh-pos-relax-4}), together with the fact that for any positive numbers $a_i$ and $b_i$,
\[\frac{\sum_i a_i}{\sum_i b_i} \leq \max_i \frac{a_i}{b_i}.\]
\hfill$\square$ \end{proof}

\subsection{Proof of Proposition \ref{prop:k-unknow-POS>1}}
\label{app:prop:k-unknow-POS>1}

\begin{proof}
Fix an arbitrary \WPSshort rule that is independent of $K_i$.
Let the base planning period $B=1$.
We construct two instance profiles with two retailers, both having $K_0 = 5$ and $H_1 = H_2 = 1$.
We set $K_1=1$ and $K_2=6$ for the first instance, and $K'_1=6$ and $K'_2=1$ for the second.
Since the two instances differ only in $K_i$, which is assumed to be private information, the weight vector~$(w_1,w_2)$ must remain the same for them.
Assume without loss of generality that $w_1 \leq w_2$, and let us consider the first instance with $K_1=1$ and $K_2=6$.
Since $\sqrt{\frac{K_0+K_1}{H_1}}=\sqrt{\frac{K_0+K_1+K_2}{H_1+H_2}}=\sqrt{6}$, we have that in the optimal centralized policy, $T^c_{1}=T^c_{2}=T^c_{\min}  \overset {POT}{=} \sqrt{6} \overset {POT}{=} 2$ and the total cost is 
\[
C(T^c)=(H_1+H_2)T^c_{\min}+\frac{K_0+K_1+K_2}{T^c_{\min}}=10.
\]

Next we apply Algorithm~1 with order $(2,1)$ to find the payoff dominant Nash equilibrium.
For retailer 2 we have 
\[
\sqrt{\frac{1}{2}\frac{K_2+w_2K_0}{H_2}} \geq\sqrt{4.25} > 2=T^c_2, 
\]
which implies that it will double its replenishment interval according to Lemma \ref{lem:right-jump}.
After that, it is easy to verify that no agent can benefit by changing its replenishment interval.
Thus, we get the payoff dominant Nash equilibrium $T^w=(T^c_{\min},2T^c_{\min})=(2,4)$ with $C(T^w)=10.5$.
% For any other Nash equilibrium $T'$, if $T'_{\min} \neq T^c_{\min}$, then the total cost is at least 11.
% It remains to consider the case when $T'_{\min}=T^c_{\min}$.
% For retailer 1, since 
% \[
% T^c_{\min}=2 > \sqrt{\frac{1}{2}\frac{K_1+K_0}{H_1}}=\sqrt{\frac{7}{2}},
% \]
% We have $T'_1 \leq T^c_{\min}$ for any Nash equilibrium.
% Then it remains to consider the case with $T'_1=T^c_{\min}$ and $T'_2 \geq 2T^c_{\min}$, where the total cost is at least 10.5.
Then the PoS of the \WPSshort rule is at least
\[
% \text{PoS} \geq 
\frac{C(T^w)}{C(T^c)}=\frac{10.5}{10}=1.05.
\]
\hfill$\square$ \end{proof}

\newpage
\section{Missing Proofs in Section \ref{sec:H-unkown}}

\subsection{Proof of Proposition \ref{prop:jump-bound}}
\label{app:prop:jump-bound}

\begin{proof}
Denote $k=\max_{i \in N}\frac{T^w_i}{T^c_i}$ and let $i_0$ be the first retailer who increases its replenishment interval by $k$ times during the process of Algorithm~1.
% By the choice of $i_0$ we have $U \subseteq N[T^w_{i_0}/2;T^w_{-i_0}]$.
% We give an upper bound for~$k$.
% Since $T^w$ is a Nash equilibrium, retailer $i_0$ cannot reduce its cost by halving its replenishment interval.
% According to Lemma \ref{lem:left-jump}, we have 
% \[
% T^w_{i_0}=kT^c_{i_0} \leq \sqrt{2\frac{K_{i_0}+\frac{w_{i_0}}{\sum_{j \in N[T^w_{i_0}/2;T^w_{-{i_0}}]}w_j}K_0}{H_{i_0}}} \leq \sqrt{2\frac{K_{i_0}+\frac{w_{i_0}}{\sum_{j \in U}w_j}K_0}{H_{i_0}}}.
% \]
Denote $N_{i_0}$ the set of retailers who have replenishment interval at most $T^w_{i_0}/2$ just before $i_0$'s last jump.
By the choice of $i_0$ we have $U \subseteq N_{i_0}$.
Since $i_0$'s last jump reduces its cost, according to Lemma \ref{lem:right-jump}, we have 
\begin{align*}
\frac{T^w_{i_0}}{2}& \leq \sqrt{\frac{1}{2}\frac{K_{i_0}+\frac{w_{i_0}}{\sum_{j \in N_{i_0}}w_j}K_0}{H_{i_0}}} 
\leq \sqrt{\frac{1}{2}\frac{K_{i_0}+\frac{w_{i_0}}{\sum_{j \in U}w_j}K_0}{H_{i_0}}}. 
\end{align*}
It follows that
\begin{align}\label{eq:wpsd-jump-bound-1}
kT^c_{i_0}=T^w_{i_0}\leq \sqrt{2\frac{K_{i_0}+\frac{w_{i_0}}{\sum_{j \in U}w_j}K_0}{H_{i_0}}}. 
\end{align}
By the definition of $\gamma_w$, we have
\begin{align}\label{eq:wpsd-jump-bound-2}
\frac{\sum_{j \in U}H_j}{H_{i_0}}\frac{w_{i_0}}{\sum_{j \in U}w_j} &\leq \frac{\sum_{j \in U} \frac{w_{i_0}}{H_{i_0}} H_j}{\sum_{j \in U}w_j} 
% \nonumber \\ & 
\leq \gamma_w \frac{\sum_{j \in U} \frac{w_j}{H_j} H_j}{\sum_{j \in U}w_j} 
% \nonumber \\ &
= \gamma_w. 
\end{align}
% Since $T^c_{\min} \overset{POT}{=} \frac{K_0+\sum_{j \in U}K_j}{\sum_{j \in U}H_j} \geq \frac{\frac{1}{2}(K_0+\sum_{j \in U}K_j}{\sum_{j \in U}H_j)}$
If $i_0 \in U$, then 
\begin{align}\label{eq:wpsd-jump-bound-U-1}
T^c_{i_0}=T^c_{\min} \geq \sqrt{\frac{1}{2}\frac{K_0+\sum_{j \in U}K_j}{\sum_{j \in U}H_j}},
\end{align}
and 
\begin{align}\label{eq:wpsd-jump-bound-U-2}
\frac{K_{i_0}}{H_{i_0}} \leq \frac{K_0+\sum_{j \in U}K_j}{\sum_{j \in U}H_j}.
\end{align}
So
\begin{align*}
k \overset{(\ref{eq:wpsd-jump-bound-1},\ref{eq:wpsd-jump-bound-U-1})}{\leq} \frac{\sqrt{2\frac{K_{i_0}+\frac{w_{i_0}}{\sum_{j \in U}w_j}K_0}{H_{i_0}}}}{\sqrt{\frac{1}{2}\frac{K_0+\sum_{j \in U}K_j}{\sum_{j \in U}H_j}}} 
&= 2\sqrt{\frac{\frac{K_{i_0}}{H_{i_0}}\sum_{j \in U}H_j+\frac{\sum_{j \in U}H_j}{H_{i_0}}\frac{w_{i_0}}{\sum_{j \in U}w_j}K_0}{K_0+\sum_{j \in U}K_j}} \\
&\leq 2\sqrt{\frac{K_{i_0}}{H_{i_0}}\frac{\sum_{j \in U}H_j}{K_0+\sum_{j \in U}K_j}+\frac{\sum_{j \in U}H_j}{H_{i_0}}\frac{w_{i_0}}{\sum_{j \in U}w_j}} \\
&\overset{(\ref{eq:wpsd-jump-bound-2},\ref{eq:wpsd-jump-bound-U-2})}{\leq}  2\sqrt{1+\gamma_w}.
\end{align*}
If $i_0 \in V$, then 
\begin{align}\label{eq:wpsd-jump-bound-V-1}
T^c_{i_0} \geq \sqrt{\frac{1}{2}\frac{K_{i_0}}{H_{i_0}}},
\end{align} and 
\begin{align}\label{eq:wpsd-jump-bound-V-2}\frac{K_{i_0}}{H_{i_0}} >\frac{K_0+\sum_{j \in U}K_j}{\sum_{j \in U}H_j} \geq \frac{K_0}{\sum_{j \in U}H_j}.
\end{align} So
\begin{align*}
k \overset{(\ref{eq:wpsd-jump-bound-1},\ref{eq:wpsd-jump-bound-V-1})}{\leq} \frac{\sqrt{2\frac{K_{i_0}+\frac{w_{i_0}}{\sum_{j \in U}w_j}K_0}{H_{i_0}}}}{\sqrt{\frac{1}{2}\frac{K_{i_0}}{H_{i_0}}}}
&= 2\sqrt{1+\frac{w_{i_0}}{\sum_{j \in U}w_j}\frac{K_0}{K_{i_0}}} \\
&\overset{(\ref{eq:wpsd-jump-bound-V-2})}{\leq} 2\sqrt{1+\frac{w_{i_0}}{\sum_{j \in U}w_j}\frac{\sum_{j \in U}H_j}{H_{i_0}}} \\
&\overset{(\ref{eq:wpsd-jump-bound-2})}{\leq} 2\sqrt{1+\gamma_w}.
\end{align*}

This finishes the proof of $\max_{i \in N}\frac{T^w_i}{T^c_i} \leq 2\sqrt{1+\gamma_w}$.
For PoS, since $T^c_i \leq T^w_i \leq 2\sqrt{1+\gamma_w} T^c_i$ for every retailer $i \in N$, it follows that the holding cost for the whole system is increased by at most $2\sqrt{1+\gamma_w}$ times while the setup cost for the whole system does not increase.
Therefore, the PoS of \WPSw is at most $2\sqrt{1+\gamma_w}$.
\hfill$\square$ \end{proof}

\subsection{Proof of Theorem~\ref{thm:general-POS}}
\label{app:thm:general-POS}

The main idea of improving the upper bound $O(\sqrt{\gamma_w})$ implied by Proposition~\ref{prop:jump-bound} to $O(\sqrt{\log \gamma_w})$ in Theorem~\ref{thm:general-POS} is that, instead of analyzing each retailer individually, we group the retailers based on their $H_i/w_i$ ratios and collectively bound the cost for retailers in the same group.
In particular, for retailers in $U$, we partition them into $m \coloneqq \lceil \log \gamma_w \rceil$ groups such that the $H_i/w_i$ values within each group differ by at most a factor of 2, and we bound the replenishment interval of all retailers within each group by the maximum interval in that group.
Based on this partition, we can prove that the total holding cost for retailers in $U$ is at most $\left(\sqrt{m}+\frac{1}{\sqrt{2}}\right)C(T^c)$. 

\begin{lemma}
\label{lem:general-POS-U}
For \WPSw, we have 
\begin{align}
\sum_{i\in U}H_iT^w_i \leq   \left(\sqrt{m}+\frac{1}{\sqrt{2}}\right)C(T^c)\label{eq:d-pos-U-bound}.
\end{align}
\end{lemma}

\begin{proof}
% According to Lemma \ref{lem:jump-bound}, we have $\frac{T_i^d}{T_i^c}=\frac{T_i^d}{T_{\min}^c} \leq 2\sqrt{1+\gamma_w}$ for each retailer $i \in U$.
% Assume $h_1 \geq h_2 \geq \dots \geq h_n$. Then $\gamma_w=\frac{h_1}{h_n}$.
Denote $\hmin=\min_{i \in U}\frac{H_i}{w_i}$. Then $\frac{H_i}{w_i} \leq \hmin \gamma_w$ for any retailer $i \in N$.
We partition all retailers from $U$ into $m \coloneqq \lceil \log \gamma_w \rceil$ groups $N_1,N_2,\dots,N_m$ according to their $\frac{H_i}{w_i}$, where 
\[
N_i =\{j \in U \mid 2^{i-1}\hmin \leq \frac{H_j}{w_j} < 2^{i}\hmin\}, \forall 1 \leq i \leq m-1,
\] 
and 
\[
N_m =\{j \in U \mid 2^{m-1}\hmin \leq \frac{H_j}{w_j} \leq 2^{m}\hmin\}.
\]
Notice that for any group $N_i$ and any two retailers $i_1,i_2 \in N_i$, we have 
\begin{align}\label{eq:d-pos-U-h<2}
\frac{1}{2} \leq \frac{H_{i_1} / w_{i_1}}{H_{i_2} / w_{i_2}} \leq 2.
% \max\left\{\frac{H_{i_1}}{w_{i_1}}/\frac{H_{i_2}}{w_{i_2}},\frac{H_{i_2}}{w_{i_2}}/\frac{H_{i_1}}{w_{i_1}}\right\}\leq2.
\end{align}
% For simplicity, we 
% denote $T(N_i)=2^{i-1}T_{\min}^c$ and 
% denote $\overline{N_i}=N_1 \cup N_2 \cup \dots \cup N_i$.
% Assume that retailers in $U$ are ordered in a way such that all retailers from $N_i$ have smaller indexes than all retailers from $N_j$ if $i < j$.
Apply Algorithm 1 with an arbitrary but fixed update order.
Denote $T(N_i)=\max_{j \in N_i} T^w_j$ the largest replenishment interval of retailers in $N_i$.
We will use $T(N_i)$ as an upper bound for the replenishment interval of all retailers in $N_i$.
For each $1 \leq i \leq m$, let $r_i \in N_i$ be the first retailer in $N_i$ who jumps to $T(N_i)$ during Algorithm 1.
Denote $\overline{N_i}$ the set of retailers who have replenishment interval at most $T(N_i)/2$ just before $r_i$'s last jump to $T(N_i)$.
By the choice of $r_i$ we have 
\begin{align}\label{eq:d-pos-U-N}
N_i \subseteq \overline{N_i}.
\end{align}
In addition, according to Lemma \ref{lem:right-jump} and retailer $r_i$'s last jump, we have 
\begin{align}\label{eq:d-pos-U-T-up}
& \frac{T(N_i)}{2} < \sqrt{\frac{1}{2}\frac{\frac{w_{r_i}}{\sum_{j \in \overline{N_i}}w_j}K_0+K_{r_i}}{H_{r_i}}}\nonumber \\
\Rightarrow \quad  &
T(N_i) < \sqrt{2\frac{\frac{w_{r_i}}{\sum_{j \in \overline{N_i}}w_j}K_0+K_{r_i}}{H_{r_i}}}.
\end{align}
% On the other hand, since all retailers from $N_i$ do not jump to $2T(N_i)$, according to Lemma \ref{lem:left-jump}, we have
% \begin{align}\label{eq:d-pos-T-dw}
% T(N_i) \geq \sqrt{\frac{1}{2}\frac{\frac{w_\ell}{\sum_{j \in \overline{N_i}}w_j}K_0+K_\ell}{H_\ell}}, \forall \ell \in N_i.
% \end{align}
% Combining inequalities (\ref{eq:d-pos-T-up}) and (\ref{eq:d-pos-T-dw}) we get
% \[
% \frac{\frac{w_\ell}{\sum_{j \in \overline{N_i}}w_j}K_0+K_\ell}{H_\ell} \leq 4\frac{\frac{w_{r_i^*}}{\sum_{j \in \overline{N_i}}w_j}K_0+K_{r_i^*}}{H_{r_i^*}}, \forall \ell \in N_i.
% \]
% Denote $H(N_i)=\sum_{j \in N_i} H_j$ for each $1 \leq i \leq m$.
Then the holding cost for retailers in $N_i$ is bounded from above by
\begin{align}
 \quad \sum_{j \in N_i} H_jT^w_j 
& \leq \sum_{j \in N_i} H_jT(N_i) \nonumber\\
& \overset{(\ref{eq:d-pos-U-T-up})}{<} \sum_{j \in N_i} H_j \sqrt{2\frac{\frac{w_{r_i}}{\sum_{j \in \overline{N_i}}w_j}K_0+K_{r_i}}{H_{r_i}}} \nonumber \\
& \leq \sum_{j \in N_i} H_j \left(\sqrt{2\frac{\frac{w_{r_i}}{\sum_{j \in \overline{N_i}}w_j}K_0}{H_{r_i}}}+\sqrt{2\frac{K_{r_i}}{H_{r_i}}}\right) \label{eq:d-pos-U-1}\\
& \leq \sum_{j \in N_i} \left(\sqrt{H_j}\sqrt{2\frac{H_j}{H_{r_i}}\frac{w_{r_i}}{\sum_{j \in \overline{N_i}}w_j}K_0}+H_j\sqrt{2\frac{K_{r_i}}{H_{r_i}}}\right) \nonumber \\
& \overset{(\ref{eq:d-pos-U-h<2})}{\leq} \sum_{j \in N_i} \left(2\sqrt{H_j}\sqrt{\frac{w_j}{\sum_{j \in \overline{N_i}}w_j}K_0}+H_j\sqrt{2\frac{K_{r_i}}{H_{r_i}}}\right) \nonumber \\
& \overset{(\ref{eq:U-V})}{\leq} \sum_{j \in N_i} \left(2\sqrt{H_j}\sqrt{\frac{w_j}{\sum_{j \in \overline{N_i}}w_j}K_0}+H_j\sqrt{2s}\right) \label{eq:d-pos-U-2} \\
& \leq 2\sqrt{H(N_i)}\sqrt{\frac{\sum_{j \in N_i}w_j}{\sum_{j \in \overline{N_i}} w_j}K_0}+H(N_i)\sqrt{2s} \label{eq:d-pos-U-3} \\
& \overset{(\ref{eq:d-pos-U-N})}{\leq} 2\sqrt{H(N_i)K_0}+H(N_i)\sqrt{2s} \nonumber,
\end{align}
where (\ref{eq:d-pos-U-1}) applies the inequality $\sqrt{a+b} \leq \sqrt{a}+\sqrt{b}$ for any $a,b \geq 0$, and (\ref{eq:d-pos-U-3}) follows by Cauchy-Schwarz inequality.
Now the holding cost of retailers from $U$ can be bounded from above by
\begin{align}
\sum_{i\in U}H_iT^w_i 
& =\sum_{i=1}^{m}\sum_{j \in N_i} H_jT^w_j \nonumber \\
& \leq \sum_{i=1}^{m}2\sqrt{H(N_i)K_0}+\sum_{i=1}^{m}H(N_i)\sqrt{2s} 
\nonumber \\
& \leq  2\sqrt{H(U)mK_0}+\sqrt{2H(U)(K_0+K(U))} \label{eq:d-pos-U-5},
\end{align}
where (\ref{eq:d-pos-U-5}) holds because of Cauchy-Schwarz inequality.
Since 
\begin{align*}
C(T^c) \geq H(U)T^c_{\min}+\frac{K_0+K(U)}{T^c_{\min}}
\geq 2\sqrt{H(U)(K_0+K(U))}, 
% \label{eq:d-pos-Tc-bound}
\end{align*}
we have
\begin{align*}
\sum_{i\in U}H_iT^w_i \leq  2\sqrt{H(U)mK_0}+\sqrt{2H(U)(K_0+K(U))} \leq \left(\sqrt{m}+\frac{1}{\sqrt{2}}\right)C(T^c).
\end{align*}
\hfill$\square$ \end{proof}

Next we consider retailers in $V$. 
For those retailers, we cannot upper bound $K_i/H_i$ by $s$ as in (\ref{eq:d-pos-U-2}).
To overcome this issue, in addition to $H_i/w_i$, we further partition retailers with similar $H_i/w_i$ based on their $K_i/H_i$ and employ a more intricate technique to bound the holding cost.

\begin{lemma}
\label{lem:general-POS-V}
For \WPSw, we have 
\begin{align}
\sum_{i\in V}H_iT^w_i \leq 3\sqrt{m}C(T^c).\label{eq:d-pos-V-bound}
\end{align}
\end{lemma}

\begin{proof}
% Next we consider retailers in $V$. 
Similarly as the analysis for $U$, we will partition retailers in $V$ according to their $H_i/w_i$. The difference is that for retailer $i \in V$ we cannot upper bound $K_i/H_i$ by $s$ as in (\ref{eq:d-pos-U-2}).
To deal with this issue, we further partition retailers by their $K_i/H_i$.
Denote $p=\lceil \log \frac{\max_{i \in V}(K_i/H_i)}{s}\rceil+1$, then $K_i/H_i<2^ps$ for any $i \in V$.
% To avoid too many groups, we first divide $V=V_1\cup V_2$, where $V_1=\{i \in V \mid s \leq \frac{K_i}{H_i} < 2^{m+5} s\}$ and $V_2=\{i \in V \mid \frac{K_i}{H_i} \geq 2^{m+5} s\}$.
% Let us first consider retailers in $V_2$.
% For any retailer $i \in V_2$ with $T^w_i > T^c_i$, denote $\overline{N_i}$ the set of retailers who have replenishment interval at most $T^w_i/2$ just before $i$'s last jump to $T^w_i$.
% By the definition of $V_2$, we have
% \[
% T^c_i \geq \sqrt{\frac{1}{2}\frac{K_i}{H_i}} \geq \sqrt{2^{m+4}s}.
% \]
% According to Lemma \ref{lem:jump-bound}, for any retailer $j \in U$, we have
% \[
% T^w_j \leq 2\sqrt{1+\gamma_w}T^c_j \leq \sqrt{8\gamma_w}T^c_j=\sqrt{8\gamma_w}T^c_{\min} \leq \sqrt{8\gamma_w}\sqrt{2s} \leq \sqrt{2^{m+4} s} \leq T^c_i.
% \]
% Thus, we have $U \subseteq \overline{N_i}$.
% Then
% \begin{align}\label{eq:d-pos-U-T-up}
% \frac{T^w_i}{2} < \sqrt{\frac{1}{2}\frac{\frac{w_i}{\sum_{j \in \overline{N_i}}w_j}K_0+K_{i}}{H_{i}}}.
% \end{align}
% According to Corollary \ref{cor:jump-bound} and $\gamma_w \leq 2^m$, for retailers in $V_2$, we have 
% \begin{align} \label{eq:d-pos-V2-bound}
% \frac{\sum_{i \in V_2} H_iT^w_i}{\sum_{i \in V_2} H_iT^c_i} \leq 2\sqrt{2}.
% \end{align}
% It remains to consider retailers in $V_1$.
We partition all retailers from $V$ into $mp$ groups $N_{1,1},N_{1,2},\dots,N_{m,p}$ according to their $h_i$ and $K_i/H_i$, where 
$N_{i,j}=\left\{\ell \in V \mid 2^{i-1}\hmin \leq \frac{H_\ell}{w_\ell} < 2^i \hmin, 2^{j-1}s \leq \frac{K_\ell}{H_\ell} < 2^j s\right\}, \forall 1 \leq i \leq m-1, \forall 1 \leq j \leq p$,
% \[
% N_{i,j}=\left\{\ell \in V \mid 2^{i-1}\hmin \leq \frac{H_\ell}{w_\ell} < 2^i \hmin, 2^{j-1}s \leq \frac{K_\ell}{H_\ell} < 2^j s\right\}, \forall 1 \leq i \leq m-1, \forall 1 \leq j \leq p,
% \]
and $N_{m,j}=\left\{\ell \in V \mid 2^{m-1}\hmin \leq \frac{H_\ell}{w_\ell} \leq 2^m \hmin, 2^{j-1}s \leq \frac{K_\ell}{H_\ell} < 2^j s\right\}, \forall 1 \leq j \leq p$.
% \[
% N_{m,j}=\left\{\ell \in V \mid 2^{m-1}\hmin \leq \frac{H_\ell}{w_\ell} \leq 2^m \hmin, 2^{j-1}s \leq \frac{K_\ell}{H_\ell} < 2^j s\right\}, \forall 1 \leq j \leq p.
% \]
% \note{It is not necessary to bound $p$ by $m$, so the following proof does not need Proposition~\ref{prop:jump-bound} anymore.}
Notice that for any group $N_{i,j}$ and any two retailers $i_1,i_2 \in N_{i,j}$, we have
\begin{align}\label{eq:d-pos-V-h<2}
\frac{1}{2} \leq \frac{H_{i_1} / w_{i_1}}{H_{i_2} / w_{i_2}} \leq 2.
% \max\left\{\frac{H_{i_1}}{w_{i_1}}/\frac{H_{i_2}}{w_{i_2}},\frac{H_{i_2}}{w_{i_2}}/\frac{H_{i_1}}{w_{i_1}}\right\}\leq2,
\end{align}
and
\begin{align}\label{eq:d-pos-V-HK<2}
\frac{1}{2} \leq \frac{K_{i_1} / H_{i_1}}{K_{i_2} / H_{i_2}} \leq 2.
% \max\left\{\frac{K_{i_1}}{H_{i_1}}/\frac{K_{i_2}}{H_{i_2}},\frac{K_{i_2}}{H_{i_2}}/\frac{K_{i_1}}{H_{i_1}}\right\}\leq 2.
\end{align}
Moreover, since $T^c_\ell\overset{POT}{=}\sqrt{K_\ell/H_\ell}$ for any $\ell \in V$ and $T^c_{\min}\overset{POT}{=}\sqrt{s}$, for any retailer $\ell \in N_{i,j}$ we have
\begin{align}\label{eq:d-pos-V-T}
2^{\lceil\frac{j-1}{2}\rceil}\leq\frac{T^c_\ell}{T^c_{\min}}\leq 2^{\lceil\frac{j}{2}\rceil}.
\end{align} 
It follows that for any two retailers $i_1,i_2 \in N_{i,j}$ we have
\begin{align}\label{eq:d-pos-V-T<2}
\frac{1}{2} \leq \frac{T^c_{i_1}}{T^c_{i_2}} \leq 2
% \max\left\{\frac{T^c_{i_1}}{T^c_{i_2}},\frac{T^c_{i_2}}{T^c_{i_1}}\right\}\leq 2.
\end{align}

% Again, We apply Algorithm 1 with an arbitrary but fixed order.
Denote $T(N_{i,j})=\max_{\ell \in N_{i,j}} T^w_\ell$ the largest replenishment interval of retailers from $N_{i,j}$.
We will use $T(N_{i,j})$ as an upper bound for the replenishment interval of all retailers in $N_{i,j}$.
For each $1 \leq i \leq m$ and $1 \leq j \leq p$, let $r_{i,j} \in N_{i,j}$ be the first retailer in $N_{i,j}$ who jumps to $T(N_{i,j})$.
Denote $\overline{N_{i,j}}$ the set of retailers who have replenishment interval at most $T(N_{i,j})/2$ just before $r_{i,j}$'s last jump to $T(N_{i,j})$.
By the choice of $r_{i,j}$ we have
\begin{align}\label{eq:d-pos-V-N}
N_{i,j} \subseteq \overline{N_{i,j}}.
\end{align}
In addition, according to Lemma \ref{lem:right-jump} and retailer $r_{i,j}$'s last jump, we have 
\begin{align}\label{eq:d-pos-V-T-up}
& \frac{T(N_{i,j})}{2} < \sqrt{\frac{1}{2}\frac{\frac{w_{r_{i,j}}}{\sum_{\ell \in \overline{N_{i,j}}}w_\ell}K_0+K_{r_{i,j}}}{H_{r_{i,j}}}} \nonumber \\
\Rightarrow \quad &
T(N_{i,j}) < \sqrt{2\frac{\frac{w_{r_{i,j}}}{\sum_{\ell \in \overline{N_{i,j}}}w_\ell}K_0+K_{r_{i,j}}}{H_{r_{i,j}}}}.
\end{align}

Now for any retailer $\ell \in N_{i,j}$, its holding cost can be upper bounded as follows:
\begin{align*}
2H_\ell T^w_\ell &\leq \frac{H_\ell {(T^w_\ell)}^2}{2T^c_{r_{i,j}}\sqrt{m}}+2\sqrt{m}H_\ell T^c_{r_{i,j}} \\
& \leq \frac{H_\ell {(T(N_{i,j}))}^2}{2T^c_{r_{i,j}}\sqrt{m}}+2\sqrt{m}H_\ell T^c_{r_{i,j}} \\
& \overset{(\ref{eq:d-pos-V-T-up})}{\leq} \frac{1}{T^c_{r_{i,j}}\sqrt{m}}\frac{H_\ell}{H_{r_{i,j}}}\left(\frac{w_{r_{i,j}}}{\sum_{\ell \in \overline{N_{i,j}}}w_\ell}K_0+K_{r_{i,j}}\right)+2\sqrt{m}H_\ell T^c_{r_{i,j}} \\
& \overset{(\ref{eq:d-pos-V-h<2})}{\leq} \frac{1}{\sqrt{m}}\left(2\frac{w_{\ell}}{\sum_{\ell \in \overline{N_{i,j}}}w_\ell}\frac{K_0}{T^c_{r_{i,j}}}+\frac{H_\ell}{T^c_{r_{i,j}}}\frac{K_{r_{i,j}}}{H_{r_{i,j}}}\right)+2\sqrt{m}H_\ell T^c_{r_{i,j}} \\
& \overset{(\ref{eq:d-pos-V-HK<2})}{\leq} \frac{1}{\sqrt{m}}\left(2\frac{w_{\ell}}{\sum_{\ell \in \overline{N_{i,j}}}w_\ell}\frac{K_0}{T^c_{r_{i,j}}}+2\frac{K_\ell}{T^c_{r_{i,j}}}\right)+2\sqrt{m}H_\ell T^c_{r_{i,j}} \\
& \overset{(\ref{eq:d-pos-V-T})}{\leq} \frac{2}{\sqrt{m}}\frac{w_{\ell}}{\sum_{\ell \in \overline{N_{i,j}}}w_\ell}\frac{K_0}{2^{\lceil\frac{j-1}{2}\rceil}T^c_{\min}}+\frac{2}{\sqrt{m}}\frac{K_\ell}{T^c_{r_{i,j}}}+2\sqrt{m}H_\ell T^c_{r_{i,j}} \\
& \overset{(\ref{eq:d-pos-V-T<2})}{\leq} \frac{2}{\sqrt{m}}\frac{w_{\ell}}{\sum_{\ell \in \overline{N_{i,j}}}w_\ell}\frac{K_0}{2^{\lceil\frac{j-1}{2}\rceil}T^c_{\min}}+\frac{4}{\sqrt{m}}\frac{K_\ell}{T^c_\ell}+4\sqrt{m}H_\ell T^c_\ell
\end{align*}
Then we can bound the holding cost for retailers in $V$ by
\begin{align}
\sum_{\ell \in V}H_\ell T^w_\ell 
& \leq \sum_{i=1}^{m}\sum_{j=1}^p\sum_{\ell \in N_{i,j}} \frac{1}{\sqrt{m}}\frac{w_{\ell}}{\sum_{\ell \in \overline{N_{i,j}}}w_\ell}\frac{K_0}{2^{\lceil\frac{j-1}{2}\rceil}T^c_{\min}}+\sum_{\ell \in V}\left(\frac{2}{\sqrt{m}}\frac{K_\ell}{T^c_\ell}+2\sqrt{m}H_\ell T^c_\ell\right) \nonumber \\
& \leq \sum_{i=1}^{m}\sum_{j=1}^p \frac{1}{\sqrt{m}}\frac{\sum_{\ell \in N_{i,j}}w_{\ell}}{\sum_{\ell \in \overline{N_{i,j}}}w_\ell}\frac{K_0}{2^{\lceil\frac{j-1}{2}\rceil}T^c_{\min}}+\sum_{\ell \in V}2\sqrt{m}\left(\frac{K_\ell}{T^c_\ell}+H_\ell T^c_\ell\right) \nonumber \\
& \overset{(\ref{eq:d-pos-V-N})}{\leq} \sum_{i=1}^{m}\sum_{j=1}^p \frac{1}{\sqrt{m}}\frac{K_0}{2^{\lceil\frac{j-1}{2}\rceil}T^c_{\min}}+\sum_{\ell \in V}2\sqrt{m}\left(\frac{K_\ell}{T^c_\ell}+H_\ell T^c_\ell\right) \nonumber \\
& \leq \frac{1}{\sqrt{m}}\frac{K_0}{T^c_{\min}} \sum_{i=1}^m \sum_{j=1}^p2^{-\lceil\frac{j-1}{2}\rceil}+\sum_{\ell \in V}2\sqrt{m}\left(\frac{K_\ell}{T^c_\ell}+H_\ell T^c_\ell\right) \nonumber \\
& \leq 3\sqrt{m}\frac{K_0}{T^c_{\min}}+\sum_{\ell \in V}2\sqrt{m}\left(\frac{K_\ell}{T^c_\ell}+H_\ell T^c_\ell\right) \nonumber \\
& \leq 3\sqrt{m}\left(\frac{K_0}{T^c_{\min}}+\sum_{\ell \in V}\left(\frac{K_\ell}{T^c_\ell}+H_\ell T^c_\ell\right)\right) \nonumber \\
&  \leq 3\sqrt{m}C(T^c),  \nonumber
\end{align}
where the last inequality follows by
\begin{align*}
C(T^c)&=\frac{K_0}{T^c_{\min}}+\sum_{\ell \in N}\left(\frac{K_\ell}{T^c_\ell}+H_\ell T^c_\ell\right)
\geq \frac{K_0}{T^c_{\min}}+\sum_{\ell \in V}\left(\frac{K_\ell}{T^c_\ell}+H_\ell T^c_\ell\right).
\end{align*}
\hfill$\square$ \end{proof}

Now we are ready to prove Theorem~\ref{thm:general-POS}.

\begin{proof}[Proof of Theorem~\ref{thm:general-POS}]
First, since $T^c \leq T^w$, we can bound the total setup costs under $T^w$ by that under $T^c$, i.e., 
\[
\frac{K_0}{T^w_{\min}}+\sum_{i \in N}\frac{K_i}{T^w_i} \leq \frac{K_0}{T^c_{\min}}+\sum_{i \in N}\frac{K_i}{T^c_i} \leq C(T^c).
\]
Combining it with (\ref{eq:d-pos-U-bound}) and (\ref{eq:d-pos-V-bound}), we have
\begin{align*}
\frac{C(T^w)}{C(T^c)}
 &\leq \frac{\sum_{i\in U}H_iT^w_i+\sum_{i \in V}H_i T^w_i}{C(T^c)} 
%  \\   & \quad 
 + \frac{\frac{K_0}{T^w_{\min}}+\sum_{i \in N}\frac{K_i}{T^w_i}}{C(T^c)} \\
 % & \leq 4\sqrt{m}+\frac{1}{\sqrt{2}}+\frac{\frac{K_0}{T^c_{\min}}+\sum_{i \in N}\frac{K_i}{T^w_i}}{C(T^c)} \\
 %  & \leq 4\sqrt{m}+\frac{1}{\sqrt{2}}+\frac{\frac{K_0}{T^c_{\min}}+\sum_{i \in N}\frac{K_i}{T^c_i}}{C(T^c)}\\
&\leq 4\sqrt{m}+\frac{1}{\sqrt{2}}+1\\
&=4\sqrt{\lceil \log \gamma_w \rceil}+\frac{1}{\sqrt{2}}+1. 
\end{align*}
\hfill$\square$ \end{proof}

\subsection{Proof of Proposition \ref{prop:h-unknow-POS>gamma}}
\label{app:prop:h-unknow-POS>gamma}

\begin{proof}
We first fix an arbitrary \WPSshort rule that is independent of $h_i$.
We then construct an instance profile with $n$ retailers, where $K_0=1$, $K_i=0$ and $d_i=2$ for each $i \in N$.
Since $h_i$ is private, the weight vector $(w_i)_{i \in N}$ must be fixed across all instances with these public parameters, regardless of $h_i$ (or $H_i$).
Assume without loss of generality that $w_1 \leq w_2 \leq \dots \leq w_n$.
Then we set $h_i=\frac{w_i}{\sum_{1 \leq j \leq i}w_j}$ for each $i \in N$.
It follows that $H_i=\frac{1}{2}h_id_i=h_i=\frac{w_i}{\sum_{1 \leq j \leq i}w_j}$ and $\gamma_d=\frac{\max h_i}{\min h_i}$.
Notice that 
\[
\sum_{j \in N} H_j=\sum_{j \in N} h_j=\sum_{j \in N} \frac{w_j}{\sum_{1 \leq i \leq j}w_i} 
\geq \sum_{j \in N} \frac{1}{j}  \geq \ln n,
\]
and
\[
\frac{\max h_i}{\min h_i}=\frac{h_1}{h_n}=\frac{\sum_{1 \leq j \leq n}w_j}{w_n} \leq n.
\]
Thus,
\begin{align}\label{eq:h-unknow-POS>gamma}
\sum_{j \in N} H_j \geq \ln n \geq \ln \frac{\max h_i}{\min h_i} = \ln \gamma_d.
\end{align}

For the optimal centralized policy $T^c$, we have $T^c_{\min} \overset {POT}{=} \sqrt{\frac{K_0}{\sum_{j \in N}H_j}}  \overset {POT}{=} \sqrt{\frac{1}{\sum_{j \in N}H_j}}$ and the total cost can be bounded by
\[
C(T^c)=\sum_{j \in N}H_jT^c_{\min}+\frac{K_0}{T^c_{\min}} \leq \frac{3}{\sqrt{2}}\sqrt{\sum_{j \in N}H_j}.
\]

Next we apply Algorithm 1 to find the payoff dominant Nash equilibrium $T^w$.
We show that $T^w_i \geq \frac{1}{\sqrt{2}}$ for every $i \in N$ by induction.
Let us start with retailer $n$.
Since $T^w$ is a Nash equilibrium, applying Lemma \ref{lem:right-jump} to retailer $n$, we have 
\begin{align*}
T^w_n &\geq \sqrt{\frac{\frac{w_n}{\sum_{j\in N[T^w_n;T^w_{-n}]}w_j}K_0}{2H_n}}
= \sqrt{\frac{\frac{w_n}{\sum_{j\in N[T^w_n;T^w_{-n}]}w_j}K_0}{2\frac{w_n}{\sum_{j\in N}w_j}}}\\
&=\sqrt{\frac{\sum_{j\in N}w_j}{2\sum_{j\in N[T^w_n;T^w_{-n}]}w_j}}
\geq \frac{1}{\sqrt{2}}.
\end{align*}
Assume that $T^w_i \geq \frac{1}{\sqrt{2}}$ holds for every $i \geq k+1$.
We show that it also holds for $i=k$.
Suppose that $T^w_k < \frac{1}{\sqrt{2}}$, then by assumption we have $N[T^w_k;T^w_{-k}] \subseteq \{1,2,\dots,k\}$.
Since $T^w$ is a Nash equilibrium, applying Lemma \ref{lem:right-jump} to retailer $k$, we have 
\begin{align*}
T^w_k &\geq \sqrt{\frac{\frac{w_k}{\sum_{j\in N[T^w_k;T^w_{-k}]}w_j}K_0}{2H_k}}
= \sqrt{\frac{\frac{w_k}{\sum_{j\in N[T^w_k;T^w_{-k}]}w_j}K_0}{2\frac{w_k}{\sum_{1 \leq j \leq k}w_j}}} \\
&=\sqrt{\frac{\sum_{1 \leq j \leq k}w_j}{2\sum_{j\in N[T^w_k;T^w_{-k}]}w_j}}
\geq \frac{1}{\sqrt{2}},
\end{align*}
which is a contradiction.
This finishes our induction.

Now we can bound the total cost of the payoff dominant Nash equilibrium as follows.
\[
C(T^w) \geq \sum_{j \in N} H_jT^w_j \geq \frac{\sum_{j \in N} H_j}{\sqrt{2}}.
\]
Therefore, the PoS of the \WPSshort rule is at least
\[
% \text{PoS} \geq 
\frac{C(T^w)}{C(T^c)} \geq \frac{\frac{\sum_{j \in N} H_j}{\sqrt{2}}}{\frac{3}{\sqrt{2}}\sqrt{\sum_{j \in N} H_j}}=\frac{\sqrt{\sum_{j \in N} H_j}}{3} \overset{(\ref{eq:h-unknow-POS>gamma})}{\geq} \frac{\sqrt{\ln\gamma_d}}{3}.
\]
\hfill$\square$ \end{proof}

\end{document}